\documentclass[final]{siamonline171218}

\usepackage{multirow}
\usepackage{tabulary}
\usepackage{colortbl}

\usepackage{lipsum}
\usepackage{amsfonts}
\usepackage{graphicx}
\usepackage{epstopdf}
\usepackage{amssymb}
\usepackage{amsmath}
\usepackage{algorithm}
\usepackage[noend]{algpseudocode}
\usepackage{enumitem}
\usepackage{diagbox}

\newcommand{\bb}[1]{\mathbb{#1}}

\newcommand{\norm}[1]{\left\lVert#1\right\rVert}
\newcommand{\Langle}{\Big\langle}
\newcommand{\Rangle}{\Big\rangle}

\newcommand{\fix}{\text{Fix}}
\newcommand{\zer}{\text{Null}}

\newcommand{\hlcell}[1]{\cellcolor{blue!25}\textbf{#1}}

\graphicspath{{Figures/}}

\makeatletter
\def\BState{\State\hskip-\ALG@thistlm}
\makeatother

\headers{Regularization by Denoising via Fixed-Point Projection (RED-PRO)}{R. Cohen, M. Elad and P. Milanfar}

\title{Regularization by Denoising via Fixed-Point Projection (RED-PRO)}

\author{Regev~Cohen\thanks{Electrical Engineering Department, Technion
Israel Institute of Technology. (\email{ regev.cohen@gmail.com}).}
\and Michael~Elad\thanks{Computer Science Department, Technion Israel Institute of Technology.
(\email{ elad@cs.technion.ac.il}).} \and Peyman~Milanfar\thanks{Google Research, Mountain View, CA 94043. (\email{peyman.milanfar@gmail.com})}}

\ifpdf
\hypersetup{
  pdftitle={Regularization by Denoising via Fixed-Point Projection (RED-PRO)},
  pdfauthor={R. Cohen, M. Elad and P. Milanfar}
}
\fi

\begin{document}

\maketitle

%=========================================================================
%=========================================================================

	\begin{abstract}
    Inverse problems in image processing are typically cast as optimization tasks, consisting of data fidelity and stabilizing regularization terms. A recent regularization strategy of great interest utilizes the power of denoising engines. Two such methods are the Plug-and-Play Prior (PnP) and Regularization by Denoising (RED). While both have shown state-of-the-art results in various recovery tasks, their theoretical justification is incomplete. %More specifically, PnP exhibits no underlying global optimization, while RED's convergence relies on differentiability and a symmetric Jacobian, excluding various advanced denoisers. 
    In this paper, we aim to bridge between RED and PnP, enriching the %theoretical 
    understanding of both frameworks.
    % First, we revisit RED and consider the case where the denoiser is non-differentiable. We show that under a certain monotonicity condition, RED algorithms minimize an inverse problem where the Rockafellar function is used as the regularization.
    Towards that end, we reformulate RED as a convex optimization problem utilizing a projection (RED-PRO) onto the fixed-point set of demicontractive denoisers. We offer a simple iterative solution to this problem, by which we show that PnP proximal gradient method is a special case of RED-PRO, while providing guarantees for the convergence of both frameworks to globally optimal solutions. In addition, we present relaxations of RED-PRO that allow for handling denoisers with limited fixed-point sets. Finally, we demonstrate RED-PRO for the tasks of image deblurring and super-resolution, showing improved results with respect to the original RED framework.           
	\end{abstract}


% REQUIRED

%=========================================================================
%=========================================================================

\section{Introduction}

Inverse problems arise in numerous fields, ranging from astrophysics and optics to signal processing, computer vision and medical imaging \cite{cohen2018doppler, cohen2018sparse, cohen2018optimized, solomon2019deep,cohen2019deep, luijten2020adaptive, chernyakova2018fourier}. Specifically to computational imaging, inverse problems relate to the task of inferring an unknown image from its corrupted measurements. Assuming a known degradation model (i.e., the forward operator), one may offer a formulation of the recovery process as an optimization problem comprising a data fidelity term. Unfortunately, this by itself is typically insufficient, as it leads to an ill-posed problem with a non-stable solution. % , i.e., the solution is severely sensitive to small data perturbations. 
To overcome this difficulty% and stabilize the solution of inverse problems
, regularization methods have been widely developed and used. A proper regularization enables a robust recovery by leveraging prior information on the unknown image, incorporated into the optimization formulation. Therefore, formulating and solving inverse problems often reduces to determining the appropriate regularization, depending on the specific application and the underlying signals in mind. 

Many regularization schemes for natural images are available, and it is beyond the scope of this paper to review this rich literature. The focus of this paper is on the fascinating recent idea that \emph{image denoisers} could be used as the mechanism behind the regularization term~\cite{romano2017little, venkatakrishnan2013plug}. Image denoising is 
the simplest known inverse problem, concerned with obtaining a clean image from its noisy measurements contaminated with additive noise, typically assumed to be zero-mean Gaussian distributed. Over the past decade, image denoising has been studied extensively, leading to a vast line of works (,e.g., \cite{tasdizen2009principal, elad2006image, dabov2007image, takeda2007kernel, mairal2009non, dong2011sparsity, dong2012nonlocally, chatterjee2011patch, burger2012image, lebrun2012secrets, lebrun2013implementation, knaus2013dual, talebi2013global, gu2014weighted, dong2015image, pierazzo2015da3d,zhang2018ffdnet,chen2016trainable,zhang2017beyond, guo2019toward, krull2019noise2void, scetbon2019deep}). This has resulted in the availability of extremely efficient and effective algorithms, achieving nearly optimal performance \cite{chatterjee2009denoising, levin2011natural, levin2012patch}. 

The first to propose leveraging the power of denoising for regularization were Venkatakrishnan \textit{et al.}, presenting their Plug-and-Play Prior (PnP) framework \cite{venkatakrishnan2013plug,chan2016plug,sreehari2016plug}. PnP relies on the alternating direction method of multipliers (ADMM) for solving general inverse problems, and their method amounts to an iterative technique that decomposes the optimization problem into a sequence of simple denoising operations, known as proximal algorithms \cite{parikh2014proximal, beck2017first}. The PnP approach replaces the proximal operators with state-of-the-art denoisers, thus, embodying implicit priors for regularizing inverse problems. This framework has gained great interest due to its success in various applications \cite{sreehari2016plug,wu2019online,rick2017one,kamilov2017plug, zhang2019deep,ahmad2019plug, xing2019plug}, and it has been extended to other proximal algorithms such as proximal gradient method (PGM) \cite{beck2009fast, parikh2014proximal}, approximate message passing (AMP) \cite{metzler2016denoising,berthier2020state,fletcher2018plug, donoho2009message,bayati2011dynamics} and half quadratic splitting \cite{zhang2017learning}. Schemes similar to PnP has been proposed in \cite{danielyan2010image} and \cite{egiazarian2015single}, where the former is based on the augmented Lagrangian method and the latter relies on the notion of Nash equilibrium.

Despite its empirical success, when the denoising engines are more general than proximal maps, PnP loses its interpretability as a minimizer of explicit objective functions. While general inversion extends beyond optimization-based solutions, the existence of an underlying energy function being minimized benefits from extensive theoretical results and adds to the quality of the resulting algorithms, since it allows to properly tune the parameters, to incorporate line search and other acceleration methods, to analyze convergence better and understand the characteristics of the solutions.    
Nonetheless, several studies have proven the fixed-point convergence of variants of PnP under various assumptions. Chan \textit{et al.} \cite{chan2016plug} proved the convergence of PnP-ADMM with increasing penalty parameter under the assumption of bounded denoisers. A similar condition was used in \cite{tirer2018image} for the analysis of a variant of PnP.
In \cite{buzzard2018plug}, PnP has been explained via the framework of consensus equilibrium, proving the convergence of PnP for nonexpansive denoisers. The latter has been the core assumption in multiple works \cite{sreehari2016plug, sun2019online, teodoro2017scene, chan2019performance, teodoro2018convergent}, which proved the convergence of several PnP techniques by reformulating them as fixed-point iterations of (averaged) nonexpansive mappings. In \cite{ryu2019plug}, the authors relaxed this assumption to a certain Lipschitz condition. Yet, none of the above proved the convergence of PnP to global minima of an explicit objective function when the denoisers extend beyond proximal operators. Recently, the authors of \cite{xu2020provable} have proven, based on \cite{gribonval2019bayesian,gribonval2011should,gribonval2020characterization}, that PnP with minimum
mean squared error (MMSE) denoisers provably converges to an explicit, possibly nonconvex, cost function. However, this work assumes that the denoiser is infinitely differentiable with a symmetric Jacobian \cite{gribonval2020characterization}. 

As a partial remedy and an enrichment to the above, Romano \textit{et al.} introduced an alternative approach called Regularization by Denoising (RED) \cite{romano2017little}. Here, a denoiser engine is utilized to form an explicit regularizer, consisting of the inner product between the unknown image and its denoising residual. Interestingly, under certain conditions, the RED prior defines a convex function whose gradient is simply given by the denoising residual itself. Thus, this gradient can be used by first-order optimization methods such as steepest decent (SD), fixed-point (FP) iteration and ADMM, leading to a  convergence to the globally optimal solution of the regularized inverse problem. 

The RED framework aims at formulating a clear and well-defined objective function, and has demonstrated state-of-the-art results in image deblurring and super-resolution, drawing a broad attention \cite{reehorst2018regularization, sun2019block, mataev2019deepred, hong2019acceleration, wu2019online, perelli2019regularized, zhang2017learning, zhang2018ffdnet}. However, RED formulation requires the denoiser to be differentiable, obey a local-homogeneity property, and have a symmetric Jacobian; conditions which are not met by various powerful denoisers \cite{reehorst2018regularization}, such as non-local means (NLM)~\cite{buades2005non}, block-matching and 3D filtering (BM3D)~\cite{dabov2007image}, trainable nonlinear
reaction-diffusion (TNRD)~\cite{chen2016trainable}, and
denoising convolutional neural network (DnCNN)~\cite{zhang2017beyond}. In such cases, RED algorithms cease to act as optimization solvers and are deprived from their convergence guarantees to global optimum. To overcome this, Reehorst and Schniter introduced a framework called score-matching by
denoising (SMD) \cite{reehorst2018regularization}, which offers a new interpretation of RED algorithms and proves their converge to a fixed-point assuming nonexpansive denoisers. In \cite{sun2019block}, a block coordinate RED algorithm was presented, including an analysis of its fixed-point convergence under the condition that the denoiser is block-nonexpansive. These works, however, did not show convergence to global minima of an explicit energy function.   

As becomes clear from the above, the convergence analysis and the study of the solutions of both PnP and RED frameworks are incomplete. In this paper, our main goal is to enrich the theoretical understanding of PnP and RED by offering a new interpretation for their regularization strategy. Another objective we pose is the formation of a theoretical bridge between RED and PnP. While both RED and PnP have been interpreted as equilibrium consensus \cite{reehorst2018regularization, ahmad2019plug,buzzard2018plug}, here, we first establish a relationship between the two frameworks from an optimization point of view. 
The main contributions of this paper are two-fold: 
\begin{itemize}[leftmargin=*]
\item We re-introduce RED via the Fixed-Point Projection (RED-PRO) strategy. More specifically, given a convex objective function $\ell(\cdot)$ and a chosen denoiser $f(\cdot)$, we formulate the following inverse problem
\begin{equation}
    \underset{\textbf{x}\in\mathbb{R}^n}{\min}\;\ell(\textbf{x})\quad\text{s.t. }\textbf{x}=f(\textbf{x}).
\end{equation}

Surprisingly, while the proposed regularization term is highly nonlinear, the problem above is convex for demicontractive denoisers with non-empty fixed-point set. We present a simple provably-convergent iterative technique to solve the resulting problem, and show that PnP proximal gradient is a special case of it. Moreover, we relate the family of demicontractive operators to previous common assumptions such as nonexpansiveness and boundedness, showing that the condition of demicontractivity covers a broader range of functions. The presented RED-PRO framework joins the PnP and RED approaches, offering an increased flexibility in choosing the denoiser, while preserving global convergence guarantees. 

\item Furthermore, we propose relaxations of  RED-PRO  that enable the use of denoisers with a narrow fixed-point set at the expense of higher computational load. Results of this strategy on the image deblurring and super-resolution problems show improved performance in comparison to the RED framework.           

\end{itemize}

\noindent In addition to the above, we complement the work in \cite{reehorst2018regularization} by considering the original RED framework where the denoiser is non-differentiable. We formulate the regularization as the Rockafellar function \cite{rockafellar2009variational}, and prove that under certain monotonicity condition on the denoiser, the proposed objective is a convex function, minimized by the RED algorithms. Thus, we provide guarantees for convergence of the RED framework to the globally optimal solution, while relieving the earlier conditions on differentiability, symmetry and homogeneity of the denoiser. As this part deviates from the main theme of this paper, it is brought in \cref{subsec:REDrevisited}.

The remainder of the paper is organized as follows. Section\,\ref{sec:preliminaries} details preliminaries of inverse problems and briefly overviews the PnP and RED approaches. In Section\,\ref{sec:demi}, we review important results of fixed-point theory and discuss demicontractive denoisers extensively. Section\,\ref{sec:REDPRO} serves as the central part of this work. %We begin with providing theoretical justifications for the RED algorithms when non-differentiable denoisers are considered.
We introduce the RED-PRO framework, which exploits the fixed-point set of denoisers to solve general structured inverse problems. We discuss the relation between the proposed scheme and the PnP approach, including convergence guarantees. Relaxed versions of RED-PRO are presented, so as to allow handling denoisers with a narrow fixed-point set. Section\,\ref{sec:results} brings experimental results on image deblurring and  super-resolution, demonstrating an improvement of the relaxed RED-PRO framework over the original RED algorithms. Finally, we conclude the paper in Section\,\ref{sec:conclusion}.

%=========================================================================
%=========================================================================

\section{Preliminaries}
\label{sec:preliminaries}

In this section we provide the groundwork for this study. First, we describe the framework of inverse problems in % the context of 
image processing, 
%. Image recovery tasks are 
formulated as optimization problems where the challenge is in determining the appropriate regularizer. Then, we discuss % two breakthrough approaches for solving inverse problems, 
PnP and RED, which utilize denoising engines for regularization.  

\subsection{Inverse Problems}
We consider the task of recovering an unknown image $\bf x$ from its corrupted measurements $\bf y$. The Bayesian maximum a posteriori (MAP) estimator seeks for a solution that maximizes the posterior conditional probability
\begin{equation}
    \hat{\bf x}_\text{MAP}\triangleq \underset{\textbf{x}\in\bb{R}^n}{\arg\max}\, P({\bf x}|{\bf y}).
\end{equation}
The latter can be simplified using Bayes's rule as follows:
\begin{align}
    \begin{split}
         \hat{\bf x}_\text{MAP}&= \underset{\textbf{x}\in\bb{R}^n}{\arg\max}\,\frac{P({\bf y}|{\bf x})P({\bf x})}{P({\bf y})} \\
         % &=\underset{\textbf{x}\in\bb{R}^n}{\arg\max}\,P({\bf y}|{\bf x})P({\bf x}) \\
         &=\underset{\textbf{x}\in\bb{R}^n}{\arg\min}\, -\log P({\bf y}|{\bf x})-\log P({\bf x}),
    \end{split}
    \label{eq:MAP}
\end{align}
where we omit $P({\bf y})$ since it is not a function of $\bf x$, and we exploit the monotonic decreasing property of $-\log(\cdot)$ to recast the estimation as a minimization problem. The log-likelihood term, denoted as $\ell({\bf x};{\bf y})\triangleq -\log P({\bf y}|{\bf x})$, describes the probabilistic relationship between the  measurements $\bf y$ and the desired image $\bf x$, assumed to be known. Typically, the likelihood alone is not sufficient and leads to an ill-posed problem for which the solution is not unique or stable. The probability distribution $P({\bf x})$ leads to a prior, denoted by $\lambda\rho({\bf x})\triangleq -\log P({\bf x})$, incorporating the statistical nature of the unknown. This regularization term stabilizes and better-conditions the optimization problem. Thus, we can rewrite \eqref{eq:MAP} as
\begin{equation}
    \hat{\bf x}_\text{MAP}=\underset{\textbf{x}\in\bb{R}^n}{\arg\min}\;\ell({\bf x};{\bf y})+\lambda\rho({\bf x}),
    \label{eq:MAPoptimization}
\end{equation}
where $\lambda\geq 0$ represents the level of confidence in the prior. We assume hereafter that the log-likelihood $\ell(\textbf{x};\textbf{y})$ is a convex, differentiable, lower semicontinuous (l.s.c) and proper function. A classic model for the measurements, considered throughout this paper, is given by
\begin{equation}
    \bf y=Hx+e,
\end{equation}
where $\bf H$ is a linear
degradation operator and $\bf e$ is a white Gaussian noise (WGN) of
variance $\sigma^2$. This leads to %the following log-likelihood
\begin{equation}
    \ell({\bf x};{\bf y})=\frac{1}{2\sigma^2}\norm{\bf Hx-y}_2^2.
    \label{eq:fidelity}
\end{equation}
Note that the noise
distribution might be different, e.g., Laplacian, Gamma-distributed, Poisson, and other noise models. In these cases, the expression for the log-likelihood above changes accordingly, differing from the $L_2$-norm.  

A challenge that remains is determining the prior to incorporate into the problem formulation. This task of choosing the appropriate $\rho({\bf x})$ has been the center of numerous studies, and  various terms have been proposed, ranging from Tikhonov smoothness \cite{golub1999tikhonov} and the classic Laplacian \cite{lagendijk2009basic}, through  
wavelet sparsity \cite{mallat1999wavelet} and total variation \cite{rudin1992nonlinear}, to patch-based GMM \cite{zoran2011learning,yu2011solving},  sparse-representation modeling \cite{bruckstein2009sparse, elad2010sparse} and recent deep-learning techniques \cite{ulyanov2018deep}. A possible answer to this challenge may lay in the solution of a special inverse problem, which is image denoising:
\begin{equation}
    \hat{\bf x}_{Denoise}=\underset{\textbf{x}\in\bb{R}^n}{\arg\min}\;\frac{1}{2\sigma^2}\norm{\bf x-y}_2^2+\lambda\rho({\bf x}).
    \label{eq:MAPdenoise}
\end{equation}
To a large extent, the removal of an additive white  Gaussian noise from an image is considered in computational imaging as a solved problem \cite{romano2017little}. In the last decade, many extremely effective image denoising algorithms have been proposed, yielding impressive results. In general, the image denoising engine is a function $f:\bb{R}^n\rightarrow\bb{R}^n$ that maps an image $\bf y$ to a different image of the same size, $\hat{\bf x}=f({\bf y})$, where ideally $\hat{\bf x}=\bf x$ -- the original noiseless image. Denoising solutions may be based on the MAP estimation \eqref{eq:MAPdenoise},  minimum mean-square-error, collaborative filtering, supervised learning, and more. 
The success in noise removal has led researchers to exploit the power of denoising engines for solving other problems. In the following, we describe two approaches, PnP \cite{venkatakrishnan2013plug} and RED \cite{romano2017little}, which utilize denoisers as implicit and explicit regularization terms respectively. When the chosen denoisers are proximal operators PnP reduces to common solvers. Thus, the PnP framework extends well-known optimization schemes, while RED presents an alternative to them. When powerful denoiser are used, both these important approaches have shown state-of-the-art performance in tasks such as unsupervised image deblurring and image super-resolution.   

%=========================================================================

\subsection{Plug and Play Prior (PnP)}

Here we briefly review the PnP approach \cite{venkatakrishnan2013plug} by Venkatakrishnan \textit{et al.}, who suggested the use of denoisers as an implicit regularization. We start with problem \eqref{eq:MAPoptimization} where we assume that $\rho(\cdot)$ is convex. The MAP solution can be obtained by two common optimization solvers, PGM and ADMM, summarized in \cref{algo:PGM} and \cref{algo:ADMM} respectively. Note that when $\alpha_k\equiv0$, \cref{algo:PGM} reduces to the standard PGM method, while the following update rule
\begin{align*}
    &t_0=1,\; t_{k+1}=\frac{1}{2}\left(1+\sqrt{1+4t_k^2}\right), \alpha_k=\frac{t_k-1}{t_{k+1}}
\end{align*}
leads to the accelerated PGM \cite{beck2009fast}.

\begin{algorithm}
\caption{PGM/APGM}
\begin{algorithmic}[1]
\item\textbf{Input:} ${\bf z}_0={\bf x}_0\in\bb{R}^n$, $\mu,\,N>0$ and $\{\alpha_k\}_{k\in\bb{N}}$.
\item \textbf{for} $k=0,1,2,...,N-1$ \textbf{do}:  
\begin{itemize}
    \item $\textbf{v}_{k+1}=\textbf{x}_k-\mu\nabla\ell(\textbf{x}_k;\textbf{y})$
    \item $\textbf{z}_{k+1}=\text{P}_{\mu\lambda \rho(\cdot)}(\textbf{v}_{k+1})$
    \item $\textbf{x}_{k+1}=\textbf{z}_{k+1}+\alpha_k(\textbf{z}_{k+1}-\textbf{z}_k)$
\end{itemize}
\item \textbf{Output:} $\textbf{x}_{k+1}$.
\end{algorithmic}
\label{algo:PGM}
\end{algorithm}

\begin{algorithm}
\caption{ADMM}
\begin{algorithmic}[1]
\item\textbf{Input:} ${\bf z}_0={\bf x}_0\in\bb{R}^n$, $\beta,\,N>0$ and $\textbf{u}_0\in\bb{R}^n$.
\item \textbf{for} $k=0,1,2,...,N-1$ \textbf{do}: 
\begin{itemize} 
    \item $\textbf{x}_{k+1}=\text{P}_{\frac{1}{\beta} \ell(\cdot)}(\textbf{z}_{k}-\textbf{u}_{k})$
    \item $\textbf{z}_{k+1}=\text{P}_{\frac{\lambda}{\beta} \rho(\cdot)}(\textbf{x}_{k+1}+\textbf{u}_{k})$
    \item $\textbf{u}_{k+1}=\textbf{u}_k+\textbf{x}_{k+1}-\textbf{z}_{k+1}$
\end{itemize}
\item \textbf{Output:} $\textbf{x}_{k+1}$.
\end{algorithmic}
\label{algo:ADMM}
\end{algorithm}

As can be seen, both techniques make use of proximal operators, defined as
\begin{equation}
    \text{P}_{g(\cdot)}(\textbf{x})\triangleq\underset{\textbf{v}\in\bb{R}^n}{\arg\min}\; \frac{1}{2}\norm{\textbf{x}-\textbf{v}}_2^2+g(\textbf{v}),
    \label{eq:prox}
\end{equation}
for a closed, proper and convex function $g:\bb{R}^n\rightarrow\bb{R}$. Notice that \eqref{eq:prox} resembles \eqref{eq:MAPdenoise}, implying that proximal operators are a specific family of denoisers. Thus, the PnP framework offers to extend \cref{algo:PGM} and \cref{algo:ADMM} by replacing the proximal operators with a general denoiser,\footnote{We use hereafter the function $f(\cdot)$ to denote a general denoiser, and we omit its dependency on the noise-level (or the denoising strength to employ). However, we should note that this parameter should be chosen carefully, as it can be critical for the performance of any of the algorithms described.} $f(\cdot)$, not necessarily variational in nature, i.e., a denoiser which is not originated from any explicit regularizer. This leads to PnP-PGM and PnP-ADMM, where the latter has been introduced in the original PnP formulation \cite{venkatakrishnan2013plug}. Note that the PnP-PGM requires the gradient of $\ell(\textbf{x};\textbf{y})$, while in general any method that evaluates the proximal operator of $\ell(\textbf{x};\textbf{y})$ (rather than its gradient), e.g. PnP-ADMM, is more computationally expensive \cite{sun2019online}.    

% \begin{algorithm}
% \caption{PnP-PGM/APGM}
% \begin{algorithmic}[1]
% \item\textbf{Input:} ${\bf z}_0={\bf x}_0\in\bb{R}^n$, $\mu>0$, $\{\alpha_k\}_{k\in\bb{N}}$ and $f(\cdot)$.
% \item \textbf{for} $k=0,1,2,...$ \textbf{do}:  
% \begin{itemize}
%     \item $\textbf{v}_{k+1}=\textbf{x}_k-\mu\nabla\ell(\textbf{x}_k;\textbf{y})$
%     \item $\textbf{z}_{k+1}=f(\textbf{v}_{k+1})$
%     \item $\textbf{x}_{k+1}=\textbf{z}_{k+1}+\alpha_k(\textbf{z}_{k+1}-\textbf{z}_k)$
% \end{itemize}
% \item \textbf{Output:} $\textbf{x}_{k+1}$.
% \end{algorithmic}
% \label{algo:P3PGM}
% \end{algorithm}

% \begin{algorithm}
% \caption{PnP-ADMM}
% \begin{algorithmic}[1]
% \item\textbf{Input:} ${\bf z}_0={\bf x}_0\in\bb{R}^n$, $\beta>0$, $\textbf{u}_0\in\bb{R}^n$ and $f(\cdot)$.
% \item \textbf{for} $k=0,1,2,...$ \textbf{do}:  
% \begin{itemize}
%     \item $\textbf{x}_{k+1}=\underset{\textbf{x}\in\bb{R}^n}{\arg\min}\;\frac{\beta}{2}\norm{\textbf{x}-\textbf{z}_k+\textbf{u}_k}_2^2+\ell(\textbf{x};\textbf{y})$
%     \item $\textbf{z}_{k+1}=f(\textbf{x}_{k+1}+\textbf{u}_k)$
%     \item $\textbf{u}_{k+1}=\textbf{u}_k+\textbf{x}_{k+1}-\textbf{v}_{k+1}$
% \end{itemize}
% \item \textbf{Output:} $\textbf{x}_{k+1}$.
% \end{algorithmic}
% \label{algo:P3ADMM}
% \end{algorithm}

Empirically, incorporating powerful denoisers (such as BM3D, TNRD and DnCNN) into the PnP framework has led to state-of-the-art results in various inverse problems. However, for general denoisers other than proximal operators, the PnP methods cannot be interpreted as optimization solvers, making it difficult to theoretically investigate the stability and uniqueness of their solutions. Several studies \cite{sun2019online, chan2016plug, buzzard2018plug, teodoro2017scene, chan2019performance, teodoro2018convergent, dong2018denoising, tirer2018image, ryu2019plug} proved the convergence of PnP methods to a fixed-point under different conditions, while the work reported in \cite{sreehari2016plug} states clear conditions for a global convergence of PnP. Yet, none of these provide an explicit expression of an objective function which is minimized. In Section\,\ref{sec:REDPRO} we remedy this by introducing an optimization technique, which PnP-PGD is a special case of. Thus, we offer a novel theoretical explanation for PnP approach.

%=========================================================================

\subsection{Regularization by Denoising (RED)}
\label{subsec:RED}

As discussed above, the PnP framework has been the first to exploit denoisers for an implicit regularization, while lacking an underlying objective function. As an alternative, Romano \textit{et al.} \cite{romano2017little} introduced a different, yet related, strategy that harnesses image denoisers, called \textit{REgularization by Denoising}. The RED framework defines the following regularization term:
\begin{equation}
    \rho_\text{RED}(\textbf{x})\triangleq\frac{1}{2}\Langle \textbf{x},\textbf{x}-f(\textbf{x})\Rangle.
    \label{eq:REDprior}
\end{equation}
This prior is an image-adaptive Laplacian whose definition is based on the denoiser of choice, $f(\cdot)$. Thus, the overall optimization problem to solve is
\begin{equation}
    \hat{\bf x}_\text{RED}=\underset{\textbf{x}\in\bb{R}^n}{\arg\min}\;\ell({\bf x};{\bf y})+\frac{\lambda}{2}\Langle \textbf{x},\textbf{x}-f(\textbf{x})\Rangle.
    \label{eq:REDoptimization}
\end{equation}
The denoiser $f(\cdot)$ is assumed to obey the following assumptions, which we refer to hereafter as the RED conditions
\begin{itemize}
    \item[(C1)] \textbf{Local Homogeneity}: $\forall \textbf{x}\in\bb{R}^n$, $f\Big((1+\epsilon)\textbf{x}\Big)=(1+\epsilon)f(\textbf{x})$ for sufficiently small $\epsilon>0$.
    \item[(C2)] \textbf{Differentiability}: The denoiser $f(\cdot)$ is differentiable where $\nabla f$ denotes its Jacobian.
    \item[(C3)] \textbf{Jacobian
Symmetry} \cite{reehorst2018regularization}: $
    \nabla f(\textbf{x})^T=\nabla f(\textbf{x}),\; \forall \textbf{x}\in\bb{R}^n$.
    \item[(C4)] \textbf{Strong Passivity}: The spectral radius the Jacobian satisfies $\eta\Big(\nabla f(\textbf{x})\Big)\leq 1$.
\end{itemize}

\noindent Interestingly, under these conditions, the regularization term \eqref{eq:REDprior} is differentiable, convex, and its gradient is given by the denoising residual $\textbf{x}-f(\textbf{x})$. Furthermore, denoting by $E_\text{RED}(\textbf{x})$ the objective function of \eqref{eq:REDoptimization},  $E_\text{RED}(\textbf{x})$ is convex whenever the log-
likelihood is convex, and it gradient is given simply as 
\begin{equation}
    \nabla E_\text{RED}(\textbf{x})=\nabla \ell(\textbf{x};\textbf{y})+\lambda\Big(\textbf{x}-f(\textbf{x})\Big).
    \label{eq:REDgradient}
\end{equation}
Based on this expression, Romano \textit{et al.} have proposed several iterative algorithms -- steepest descent, fixed-point iteration and ADMM, which are guaranteed to converge to the global optimum of \eqref{eq:REDoptimization}.

In a later work \cite{reehorst2018regularization}, the authors have argued that many popular denoisers lack symmetric Jacobian (C3),\footnote{Indeed, \cite{reehorst2018regularization} was the first work to draw attention to the need for symmetry of the Jacobian.} making the gradient expression \eqref{eq:REDgradient} invalid. Moreover, they have proven that when the denoiser $f(\textbf{x})$ fails to satisfy condition (C3), there is no regularizer $\rho(\textbf{x})$ whose gradient is the denoising residual $\textbf{x}-f(\textbf{x})$. Yet, in practice, the RED algorithms converge and have demonstrated state-of-the-art results in super-resolution and image deblurring even when used with denoisers that are not differentiable, let alone exhibit symmetric Jacobian. Therefore, the theoretical justification of the RED approach remains an open question, which we aim to address in \cref{subsec:REDrevisited}.

%=========================================================================
%=========================================================================
\section{Demicontractivity}
\label{sec:demi}
In this section, we introduce our assumption of demicontractive denoisers which is the fundamental building block of our framework formulation.
We first outline basic definitions and facts of nonlinear analysis, followed by an extensive discussion on demicontractivity and its relation to other assumptions, motivating our work.   

\subsection{Fixed-Point Theory}
\label{subsec:review}
Here we review key concepts of fixed-point theory \cite{rockafellar2009variational, bauschke2011convex} on which we base our contributions. We start with considering a nonlinear mapping $T:\bb{R}^n\rightarrow\bb{R}^n$, and we say a point $\textbf{x}\in\bb{R}^n$ is a fixed-point of $T$ iff $T(\textbf{x})=\textbf{x}$. We define the the set of all fixed-points of $T$ as
\begin{equation}
    \fix(T)\triangleq \{\textbf{x}\in\bb{R}^n\,:\;T(\textbf{x})=\textbf{x}\}.
\end{equation}
Throughout the paper we assume that $\fix(T)$ is nonempty.\footnote{A reasonable assumption is that $f(0)=0$ as any denoiser of the form $f(\textbf{x})=W(\textbf{x})\textbf{x}$ satisfies this. Moreover, all the denoisers we experiment with in this work meet this condition.} Our study focuses on the set $\fix(T)$ and its favorable properties for the family of demicontractive mappings defined next.   
% \begin{definition}[Demiclosedness]
% The mapping $T$ is said to be \textit{demiclosed} at 0, if for any sequence $\{\textbf{x}_i\}$ that satisfies $\textbf{x}_i\underset{i\rightarrow\infty}{\rightarrow}\textbf{x}$ and $\textbf{x}_i-T(\textbf{x}_i)\underset{i\rightarrow\infty}{\rightarrow}0$, we have $\textbf{x}=T(\textbf{x})$.
% \end{definition}

\begin{definition}[Demicontractive]
The mapping $T$ is \textit{demicontractive} with a constant $d\in[0,1)$ (or $d$-demicontractive) if for any $\textbf{x}\in\bb{R}^n$ and $\textbf{z}\in\fix(T)$ it holds that
\begin{equation}
    \norm{T(\textbf{x})-\textbf{z}}^2\leq \norm{\textbf{x}-\textbf{z}}^2+d\norm{T(\textbf{x})-\textbf{x}}^2,
    \label{eq:demi}
\end{equation}
or equivalently
\begin{equation}
    \frac{1-d}{2}\norm{\textbf{x}-T(\textbf{x})}^2\leq\Langle \textbf{x}-T(\textbf{x}),\,\textbf{x}-\textbf{z}\Rangle.
    \label{eq:demi2}
\end{equation}
\end{definition}

\begin{definition}
The mapping $T$ is \textit{quasi-nonexpansive} if 
\begin{equation}
    \norm{T(\textbf{x})-\textbf{z}}\leq \norm{\textbf{x}-\textbf{z}},\;\forall \textbf{x}\in\bb{R}^n,\,\textbf{z}\in\fix(T).
\end{equation}
We say $T$ is $\gamma$-\textit{strongly quasi-nonexpansive} with $\gamma\geq 0$ if
\begin{equation}
    \norm{T(\textbf{x})-\textbf{z}}^2\leq \norm{\textbf{x}-\textbf{z}}^2-\gamma\norm{T(\textbf{x})-\textbf{x}}^2,\;\forall \textbf{x}\in\bb{R}^n,\,\textbf{z}\in\fix(T)..
\end{equation}
\end{definition}

\begin{proposition}
Let $T$ be a $\gamma$-strongly quasi-nonexpansive mapping with $\gamma\geq1$. Then, 
\begin{equation}
    \norm{\textbf{x}-T(\textbf{x})}\leq \norm{\textbf{x}-P_{\fix(T)}(\textbf{x})},
\end{equation}
where $P_{\fix(T)}$ represents the projection onto the fixed-point set of $T$.
\label{prop:displacemenbound}
\end{proposition}

\begin{definition}
The mapping $T$ is called \textit{nonexpansive} if
\begin{equation}
    \norm{T(\textbf{x})-T(\textbf{z})}\leq \norm{\textbf{x}-\textbf{z}},\forall \textbf{x},\textbf{z}\in\bb{R}^n.
    \label{def:nonexpansive}
\end{equation}
\end{definition}

\begin{definition}
The mapping $T$ is \textit{Lipschitz continuous} with a constant $L>0$ if
\begin{equation}
    \norm{T(\textbf{x})-T(\textbf{z})}\leq L\norm{\textbf{x}-\textbf{z}},\forall \textbf{x},\textbf{z}\in\bb{R}^n.
\end{equation}
When $L<1$, $T$ is called a contraction.
\end{definition}
\noindent Notice that any Lipschitzian mapping admits a nonexpansive function by an appropriate scaling. In addition, as illustrated in Fig.\,\ref{fig:hierarchy}, demicontractive mappings include the class of quasi-nonexpansive mappings, which in turn contains the widely studied class of nonexpansive operators. Thus, the class of demicontractive operators cover a large extent of functions and it is one of the most general classes for which some iterative methods were
investigated \cite{muarucster2011strong}, explaining their centrality in this work. Below we provide results concerning the structure of the fixed point set for demicontractive mappings. 

\begin{figure}
    \centering
    \includegraphics[trim={1cm 0cm 1cm 0cm},clip,height = 5cm, width = 0.5\linewidth]{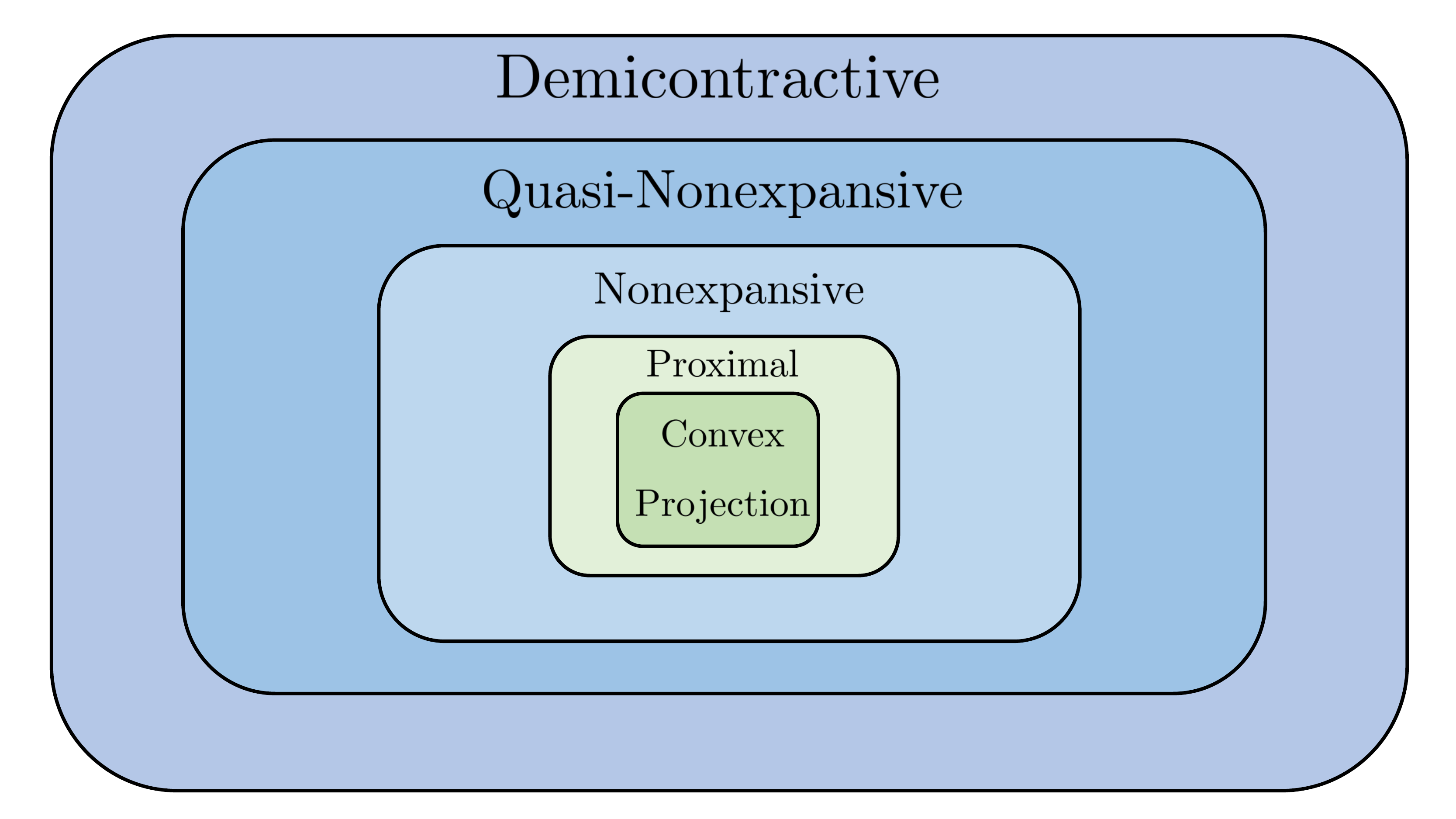}
    \caption{The family of demicontractive mappings and its subclasses.}
    \label{fig:hierarchy}
\end{figure}

\begin{definition}
Consider a mapping $T$ and let $\alpha\in(0,1)$. Then, the relaxation of $T$ is defined as the averaged  operator $T_\alpha\triangleq \alpha T +(1-\alpha)Id$ where $Id$ is the identity operator. Notice it holds that $\fix(T_\alpha)=\fix(T)$. 
\end{definition}

\begin{proposition}
Consider a $d$-demicontractive mapping $T$ and $\alpha\in (0,1-d]$. Then, the relaxed operator $T_\alpha$ is $\gamma$-strongly quasi-nonexpansive with $\gamma=\frac{(1-d-\alpha)}{\alpha}$.
\label{prop:demi2quasi}
\end{proposition}

\begin{proof}
See \cref{app:demi2quasi}.
\end{proof}

% \begin{theorem}
% Let $T_1$ and $T_2$ be two strongly quasi-nonexpansive mappings with constants $\gamma_1$ and $\gamma_2$ respectively such that $\fix(T_1)\cap\fix(T_2)\neq \emptyset$. Then, the composition $T\triangleq T_1 \circ T_2$ is $\gamma$-strongly quasi-nonexpansive where $\gamma=\frac{\gamma_1\gamma_2}{\gamma_1+\gamma_2}$ and it satisfies $\fix(T)=\fix(T_1)\cap\fix(T_2)$. 
% \label{theo:composition}
% \end{theorem}

% \begin{proof}
% See Proposition 1(d) in \cite{yamada2004hybrid}.
% % See Appendix \ref{app:composition}.
% \end{proof}

\noindent \cref{prop:demi2quasi} provides us a simple tool to construct a strong quasi-nonexpansive map from a demicontractive mapping, which we utilize in the next section.   

The following theorem is the foundation on which we shall base our problem formulation, introduced in Section\,\ref{sec:REDPRO}.
\begin{theorem}[\cite{chidume2010iterative,chidume2010iterative}]
Suppose that $T$ is a $d$-demicontractive mapping. Then, the fixed-point set $\fix(T)$ is closed and convex .
\label{theo:closed&convex}
\end{theorem}

\begin{proof}
See Lemma 5 of \cite{chidume2010iterative}.
\end{proof}

%=======================================================

\subsection{Demicontractive Denoisers}
\label{subsec:demi}
A fair and necessary question is whether general denoisers are demicontractive. While we cannot fully answer this question, we provide below supporting claims to our assumption of demicontractivity of the denoiser. We start with a typical assumption in convex optimization and gradually relax it to demicontractivity. Then, we directly relate the latter assumption to other common conditions for the convergence of PnP and RED.  

Given a denoiser $f(\cdot)$, we define the denoising residual as $r(\textbf{x})\triangleq \textbf{x}-f(\textbf{x})$. The RED framework aims at showing that under certain assumption on $f(\cdot)$, the residual $r(\cdot)$ is the gradient of an explicit convex function. Similarly, PnP can be seen an extension of Bayesian regularization \cite{sreehari2016plug} where the denoiser is a proximal operator of a (possibly implicit) function $\rho(\cdot)$, which in turn implies that the residual $r(\cdot)$ is the gradient of the convex Moreau envelope of $\rho(\cdot)$ \cite{moreau1965proximite}. Therefore, by the Baillon-Haddad theorem \cite{baillon1977quelques, bauschke2009baillon}, if the residual $r(\cdot)$ is $L$-Lipschitz continuous, it is $\frac{1}{L}$-co-coercive
\begin{equation}
    \frac{1}{L}\norm{r(\textbf{x})-r(\textbf{z})}^2\leq \langle r(\textbf{x})-r(\textbf{z}),\;\textbf{x}-\textbf{z}\rangle,\;\forall \textbf{x},\textbf{z}\in\mathbb{R}^n.
    \label{eq:cocoercive}
\end{equation}
The notion of co-coercivity plays an important role in the convergence of iterative schemes \cite{zhu1996co}. However, for general denoisers we cannot assume that the residual $r(\cdot)$ is the gradient of some convex function. Fortunately, this condition can be relaxed as follows. 
\begin{proposition}
Assume the mapping $T(\textbf{x})\triangleq r(\textbf{x})-\alpha \textbf{x}$ is Lipschitz continuous with Lipschitz constant $0<\beta\leq\alpha$, then, $r(\cdot)$ is co-coercive.
\end{proposition}

\begin{proof}
See \cite{zhu1996co} Proposition 2.3.
\end{proof}
The latter proposition provides a condition for the co-coercivity of the residual regardless if it the gradient of the a convex function or not. We further relax our assumptions on the denoisers and require that \cref{eq:cocoercive} hold only with respect to $\textbf{z}\in\zer(r)\triangleq\{\textbf{x}\in\bb{R}^n\,:\;r(\textbf{x})=\textbf{0}\}$:   
\begin{equation}
    \frac{1}{L}\norm{r(\textbf{x})-r(\textbf{z})}^2\leq \langle r(\textbf{x}),\;\textbf{x}-\textbf{z}\rangle,\;\forall \textbf{x}\in\mathbb{R}^n,\;\textbf{z}\in\zer(r). 
    \label{eq:co2demi}
\end{equation}
Substituting $r(\textbf{x})=\textbf{x}-f(\textbf{x})$ and noticing that $\zer(r)\equiv\fix(f)$, we observe that \cref{eq:co2demi} coincides with \cref{eq:demi2} with $d=1-\frac{2}{L}$. Thus, \cref{eq:co2demi} provide us an approximation of $d$, assuming we have an estimation of the Lipschitz constant of the residual, and more importantly, shows that the assumption of demicontractive denoisers is broader than co-coercivity, allowing us to go beyond common optimization schemes.    

We cannot conclude this part without discussing other previously-made assumptions and show their relation to demicontractivity. First, notice that any denoiser $f(\cdot)$ which meets conditions (C1)-(C4) is nonexpansive:
\begin{align*}
    \norm{f(\textbf{x})-f(\textbf{z})}&=\norm{\int_0^1 \nabla f\Big(\textbf{z}+t(\textbf{x}-\textbf{z})\Big)dt(\textbf{x}-\textbf{z})}\leq \norm{\int_0^1 \nabla f\Big(\textbf{z}+t(\textbf{x}-\textbf{z})\Big)dt}\cdot\norm{(\textbf{x}-\textbf{z})} \\
    &\leq \int_0^1 \norm{\nabla f\Big(\textbf{z}+t(\textbf{x}-\textbf{z})\Big)}dt\cdot\norm{(\textbf{x}-\textbf{z})} \leq \int_0^1 1dt\cdot\norm{(\textbf{x}-\textbf{z})} =\norm{(\textbf{x}-\textbf{z})}.
\end{align*}
In \cite{sun2019block}, the authors prove that the RED algorithms converge when the denoiser is block-nonexpansive. In partiuclar, when the block is taken to be the entire image, it implies that the denoiser is nonexpansive. Various studies \cite{sreehari2016plug,buzzard2018plug,sun2019online, teodoro2018convergent, teodoro2017scene, chan2019performance,reehorst2018regularization} have proven the convergence of variants of PnP for nonexpansive or averaged denoisers. Thus, all mentioned works assume demicontractive denoisers. 
Recently, the authors of \cite{ryu2019plug} have shown that PnP-ADMM and PnP forward-backward splitting converge under weaker conditions where $f(\cdot)$ Lipschitz continuous mapping satisfying
\begin{equation}
    \norm{f(\textbf{x})-f(\textbf{z})}^2\leq (1+\epsilon^2)\norm{\textbf{x}-\textbf{z}}^2,\;\forall \textbf{x},\textbf{z}\in\bb{R}^n,
\end{equation}
for some small $\epsilon>0$. Assuming $\textbf{z}\in\fix(f)$, the above condition is relaxed to the following assumption
\begin{equation}
    \norm{f(\textbf{x})-\textbf{z})}^2\leq (1+\epsilon^2)\norm{\textbf{x}-\textbf{z}}^2,\;\forall \textbf{x}\in\bb{R}^n,\textbf{z}\in\fix(f).
    \label{eq:almostnonexpansive}
\end{equation}
By the definition of demicontractivity \eqref{eq:demi2} in conjunction with the Cauchy-Schwartz inequality we obtain that any $d$-demicontractive function $f(\cdot)$ satisfies $\norm{\textbf{x}-f(\textbf{x})}^2\leq\frac{4}{(1-d)^2}\norm{\textbf{x}-\textbf{z}}^2$, which it turn implies that 
\begin{equation}
    \norm{f(\textbf{x})-\textbf{z})}^2\leq \Big(1+\frac{4d}{(1-d)^2}\Big)\norm{\textbf{x}-\textbf{z}}^2,\;\forall \textbf{x}\in\bb{R}^n,\textbf{z}\in\fix(f).
\end{equation}
Thus, it is clear that any $d$-demicontractive function meets condition \eqref{eq:almostnonexpansive} with $\epsilon^2\triangleq\frac{4d}{(1-d)^2}$.

In \cite{sun2019online}, Chan \textit{et al.} have proposed a version of PnP-ADMM for bounded denoisers 
\begin{equation}
    \frac{1}{n}\norm{f_{\sigma}(\textbf{x})-\textbf{x}}^2\leq \sigma^2c,
    \label{eq:bounded}
\end{equation}
where they assume any point $\textbf{x}$ is bounded in some interval $\textbf{x}\in[a,b]^n$ ($a<b$), $\sigma>0$ is a parameter controlling the strength of the denoiser and $c>0$ is a constant independent of $n$ and $\sigma$. Bounded denoisers are asymptotically invariant in the sense that $f_{\sigma}\rightarrow Id$ as $\sigma\rightarrow 0$. Under the boundedness assumption \eqref{eq:bounded}, a fixed-point convergence of the modified PnP-ADMM has been proven where a diminishing step size $\mu_k$ has been used and $\sigma^2_k=\lambda\mu_k$ for some predefined $\lambda>0$.\footnote{The original formulation in \cite{sun2019online} defines an increasing penalty parameter $\beta_k$ which satisfies $\beta_k=\frac{1}{\mu_k}$.} However, this approach requires the denoiser to have an internal parameter $\sigma^2$ controlling its strength which may not be available. We offer an external control which relies only on the demicontractivity of the denoiser, as given the in the following proposition.
\begin{proposition}
Consider a $d$-demicontractive denoiser $f(\cdot)$ and define the relaxed operator $f_\alpha(\cdot)$ for some $\alpha\in(0,1-d)$. Then, it holds that
\begin{equation}
    \frac{1}{n}\norm{f_\alpha(\textbf{x})-\textbf{x}}^2\leq \sigma^2(\alpha)c,
\end{equation}
where $\sigma^2(\alpha)\triangleq\frac{\alpha}{1-d-\alpha}$ and $c\triangleq(b-a)^2$.
\label{prop:demi2bound}
\end{proposition}

\begin{proof}
See \cref{app:demi2bound}
\end{proof}

\noindent The above theorem provides the mean to externally control the strength of 
a denoiser by performing appropriate averaging. This is important by itself, since it may boost the performance of PnP methods \cite{xu2020boosting}.
% Note that $\sigma^2(\alpha)\rightarrow0$ and $f_\alpha(\textbf{x})\rightarrow Id$ as $\alpha\rightarrow0$. 
Moreover, by updating $\alpha_k=\frac{\lambda\mu_k}{1+\lambda\mu_k}(1-d)\in(0,1-d)$ at each iteration we obtain $\sigma^2_k\triangleq\sigma^2(\alpha_k)=\lambda\mu_k$, which ensures the convergence of the PnP-ADMM presented in \cite{sun2019online} for demicontractive denoisers.

Following the discussion above, while demicontractivity is hard to verify for general denoisers (as other conditions), we have shown that the condition of demicontractivity is covers a wide range of operators and is and broader than other common assumptions.

%=======================================================

\section{RED-PRO: RED via Fixed-Point Projection}
\label{sec:REDPRO}
As discussed in Section\,\ref{sec:preliminaries}, the PnP and RED frameworks have achieved state-of-the-art performance in solving inverse problems by utilizing denoisers as regularization. However, their theoretical analysis is incomplete since it is unclear which objective functions are minimized by PnP and RED,\footnote{We are referring to the case where the denoiser is non-differentiable or having a non-symmetric Jacobian.} and whether such cost functions exist. 
As a partial answer, we consider in \cref{subsec:REDrevisited} non-differentiable denoisers and we provide convergence guarantees for the RED algorithms. In this section we address the matter of the underlying objective function for PnP and RED. To that end, we reformulate RED as a convex minimization problem regularized using the fixed-point set of a demicontractive denoiser. We provide simple solutions for the proposed problem, similar to PnP-PGM and PnP-ADMM. As such, we offer a theoretical explanation for the PnP approach and establish its connection to RED. Finally, we relax our problem by considering broader and richer domains than the fixed-point set. These modifications offer larger flexibility, allowing for a broader group of denoisers to be applicable.

\subsection{Regularization by Projection}
\label{subsec:REDPRO}
% As discussed earlier, several previous works rely on the nonexpansivity of the denoiser for proving the fixed-point convergence of the RED algorithms. However, the latter requirement may be too restrictive since several denoisers, e.g. NLM, are expansive \cite{chan2016plug}. To overcome the above limitation, we require henceforth a considerably weaker assumption where the denoiser $f(\cdot)$ is a $d$-demicontractive mapping for some $d\in[0,1)$.
The new framework we propose builds upon the fixed-point set of demicontractive denoisers as prior for general inverse problems. To motivate our regularization strategy, described later, we make the following observation:

\begin{proposition}
Consider a demicontractive denoiser $f(\cdot)$ and assume $f(0)=0$. Then, 
\begin{equation*}
    \rho_\text{RED}(\textbf{x})=\frac{1}{2}\Langle \textbf{x},\textbf{x}-f(\textbf{x})\Rangle=0 \;\;\text{iff}\;\; {\bf x}\in\fix(f).
\end{equation*} 
\end{proposition}

\begin{proof}
It is clear that for any ${\bf x}\in\fix(f)$, we get $\rho_\text{RED}(\textbf{x})=0$. For the other direction, we recall that by the definition of demicontractive mapping for any ${\bf x}\in\bb{R}^n$
\begin{equation}
    0\leq\frac{1-d}{2}\norm{\textbf{x}-f(\textbf{x})}^2\leq\Langle \textbf{x},\textbf{x}-f(\textbf{x})\Rangle,
\end{equation}
where we use the assumption that ${\bf x}=0$ is a fixed-point. Therefore, when $\rho_\text{RED}(\textbf{x})=0$, then $\norm{\textbf{x}-f(\textbf{x})}^2=0$, implying that ${\bf x}\in\fix(f)$.
\end{proof}

\noindent Inspired by this observation,\footnote{The assumption $f(0)=0$ is not necessary for our derivations but is only used to explain our motivation.} we introduce the following general minimization problem, which is the main message of this paper:
\begin{align}
    \begin{split}
            \hat{\bf x}_\text{RED-PRO}=\underset{\textbf{x}\in\bb{R}^n}{\arg\min}\;\, \ell(\textbf{x};\textbf{y})\;\,
    s.t.\;\, \textbf{x}\in\fix(f).
    \end{split}
    \label{eq:REDPRO}
\end{align}
We refer to the above as the RED via Fixed-Point Projection (RED-PRO) paradigm, where we utilize the fixed-point set of a denoising engine as a regularization for our inverse problem. The optimization task \eqref{eq:REDPRO} can be interpreted as searching for a minimizer of $\ell(\textbf{x};\textbf{y})$ over the set of "clean" images. Ideally, we would like to limit our search to the manifold of natural images $\mathcal{M}$ \cite{menon2020pulse,ulyanov2018deep,rick2017one}. However, the set $\mathcal{M}$ is generally not well-defined, it is not easy accessible and it is not convex,\footnote{The convex interpolation of two natural images may not be a natural image} making the search within this domain difficult. Therefore, as an alternative, we propose to use $\fix(f)$ which is well-behaved for demicontractive denoisers and should satisfy $\mathcal{M}\subset\fix(f)$ for a \textquotedblleft perfect\textquotedblright\, denoiser. Note, however, that common denoisers are far from being ideal, hence, the solution of \cref{eq:REDPRO} is sensitive to the choice of the denoiser and it may vary considerably for different choices, as shown in Section\,\ref{sec:results}.   

Surprisingly, although the constraint of \cref{eq:REDPRO} is not linear in nature, the RED-PRO approach admits a convex optimization problem as stated by the next theorem.   
\begin{theorem}
Assume the denoiser $f(\cdot)$ is a $d$-demicontractive mapping. Then, \eqref{eq:REDPRO} defines a convex minimization problem.
\label{theo:REDPROconvex}
\end{theorem}

\begin{proof}
By \cref{theo:closed&convex}, the set $\fix(f)$ is closed and convex, hence, the convexity of log-likelihood implies problem \eqref{eq:REDPRO} is a minimization of a convex function over a convex domain. 
\end{proof}

\noindent The above theorem allows to find solutions to \cref{eq:REDPRO} and to derive convergence guarantees using tools from convex optimization. Moreover, as shown in \cref{subsec:review}, it holds that $\fix(f_\alpha)=\fix(f)$ for any $\alpha\in(0,1)$ where $f_\alpha=\alpha f+(1-\alpha)Id$ is the relaxed operator. Thus, we can rewrite \eqref{eq:REDPRO} equivalently as   
\begin{equation}
             \hat{\bf x}_\text{RED-PRO}=\underset{\textbf{x}\in\bb{R}^n}{\arg\min}\;\, \ell(\textbf{x};\textbf{y})
  \;\, s.t.\; \textbf{x}\in\fix(f_\alpha).
\end{equation}
While the latter seems as an unnecessary and redundant step, it plays a crucial role in our convergence results since $f_\alpha$ is strongly quasi-nonexpansive for an appropriate choice of $\alpha$. 
To find a solution for our proposed problem \eqref{eq:REDPRO}, one may apply the projected gradient descent method \cite{parikh2014proximal, bertsekas1997nonlinear} as follows
\begin{equation}
    \textbf{x}_{k+1}=P_{\fix(f)}\Big(\textbf{x}_k-\mu_k\nabla\ell(\textbf{\textbf{x}};y)\Big),
\end{equation}
where $\mu_k>0$ is the gradient step size and $P_{\fix(f)}$ denotes the projection onto the fixed-point set $\fix(f)$. 
The above update rule resembles other studies that aim at projecting onto the set of natural images (see e.g. \cite{menon2020pulse,ulyanov2018deep,rick2017one}), however, we utilize $\fix(f)$ which exhibits favorable properties for demicontractive denoisers. Furthermore, we offer here a simpler solution, based on the hybrid steepest descent method (HSD) \cite{deutsch1998minimizing,yamada1998quadratic,bauschke2011fixed,yamada2004hybrid}, which performs one activation of the denoising engine per iteration, as detailed in \cref{algo:REDPROHSDM}.

\begin{algorithm}
\caption{HSD}
\begin{algorithmic}[1]
\item\textbf{Input:} ${\bf x}_0\in\bb{R}^n$, $\{\mu_k\}_{k\in\bb{N}}$, $\alpha\in(0,1)$, $N>0$ and $f(\cdot)$.
\item \textbf{for} $k=0,1,2,...,N-1$ \textbf{do}:  
\begin{itemize}
    \item $\textbf{v}_{k+1}=\textbf{x}_k-\mu_k\nabla\ell(\textbf{x}_k;\textbf{y})$
    \item $\textbf{z}_{k+1}=f(\textbf{v}_{k+1})$
    \item $\textbf{x}_{k+1}=(1-\alpha)\textbf{v}_{k+1}+\alpha\textbf{z}_{k+1}$
\end{itemize}
\item \textbf{Output:} $\textbf{x}_{k+1}$.
\end{algorithmic}
\label{algo:REDPROHSDM}
\end{algorithm}

\noindent We note that \cref{algo:REDPROHSDM} can be written in a compact form as
\begin{equation}
    \textbf{x}_{k+1}=f_\alpha\Big(\textbf{x}_k-\mu_k\nabla\ell(\textbf{x}_k;\textbf{y})\Big).
    \label{eq:HSDcompact}
\end{equation} 
The next theorems provide conditions for the convergence of \cref{algo:REDPROHSDM} to an optimal solution of \eqref{eq:REDPRO} where we distinguish between two cases: diminishing and constant step sizes.

\begin{theorem}[Diminishing Step Size]
\label{theo:diminishing}
Let $f(\cdot)$ be a continuous $d$-demicontractive denoiser and $\ell(\cdot;\textbf{y})$ be a proper convex l.s.c differentiable function with  L-Lipschitz gradient $\nabla\ell(\cdot;\textbf{y})$. Assume the following:
\begin{enumerate}[leftmargin=1cm]
    \item[(A1)] $\alpha\in(0,\frac{1-d}{2})$.
    \item[(A2)] $\{\mu_k\}_{k\in\bb{N}}\subset [0,\infty)$ where $\mu_k\underset{k\rightarrow\infty}{\rightarrow}0$ and $\sum_{k\in\bb{N}}\mu_k=\infty$.
\end{enumerate}
Then, the sequence $\{\textbf{x}_k\}_{k\in\bb{N}}$ generated by \cref{algo:REDPROHSDM} converges to an optimal solution of the RED-PRO problem.
\end{theorem}

\begin{proof}
See \cref{app:diminishing}.
\end{proof}

\begin{theorem}[Constant Step Size]
\label{theo:constant}
Let $f(\cdot)$ be a continuous $d$-demicontractive denoiser and $\ell(\cdot;\textbf{y})$ is a proper convex l.s.c differentiable function with L-Lipschitz gradient $\nabla\ell(\cdot;\textbf{y})$. Assume the following:
\begin{enumerate}[leftmargin=1cm]
    \item[(H1)] $\alpha\in(0,\frac{1-d}{2})$.
    \item[(H2)] $\mu_k\equiv\mu\in (0,\,\frac{2}{L})$.
    \item[(H3)] $\fix(f)\cap\fix(G_\ell)\neq\emptyset$ where $G_\ell\triangleq Id-\mu\nabla\ell$.
\end{enumerate}
Then, the sequence $\{\textbf{x}_k\}_{k\in\bb{N}}$ generated by \cref{algo:REDPROHSDM} converges an optimal solution of the RED-PRO problem.
\end{theorem}

\begin{proof}
See \cref{app:constant}.
\end{proof}

\cref{theo:diminishing} and \cref{theo:constant} provide convergence guarantees of HSD to a solution of the RED-PRO formulation. Moreover, recalling \cref{eq:HSDcompact}, \cref{algo:REDPROHSDM} resembles the PnP-PGD method, suggesting that the above theorems may provide insights into the solutions of PnP-PGD. We establish this connection in the following subsection.     

\subsubsection{PnP and RED as a special case}
PnP and RED can be seen as two alternative approaches for utilizing denoising engines as regularization for general inverse problems. Below, we aim at bridging PnP and RED via the RED-PRO framework, hopefully enriching the theoretical understanding of these two approaches and related ones.

Several variants of the RED algorithms have been proposed, see \cite{sun2019block, ahmad2019plug, reehorst2018regularization} to name just a few. In particular, the authors of \cite{ahmad2019plug} discussed an accelerated gradient descent version of RED with the following update rule
\begin{itemize}
    \item $\textbf{v}_{k+1}=\textbf{x}_k-\mu\nabla\ell(\textbf{x}_k;\textbf{y})$,
    \item $\textbf{z}_{k+1}=\textbf{v}_{k+1}+q_k(\textbf{v}_{k+1}-\textbf{v}_k)$,
    \item $\textbf{x}_{k+1}=(1-\alpha)\textbf{z}_{k+1}+\alpha f(\textbf{z}_{k+1})$,
\end{itemize}
where $q_k\geq0$ is an acceleration step size and $\alpha>0$ is a design parameter. Thus, when we set $q_k\equiv0$, i.e. when we skip the acceleration step, the above RED variant reduces to the iterative update \cref{eq:HSDcompact}. In addition, when we continue and set $\alpha=1$, we obtain the PnP-PGD method, showing that the three frameworks coincide under this setup. We note that under the assumptions of the RED-PRO framework, it is possible to set $\alpha=1$ when the denoiser in use is strongly quasi-nonexpansive, rather than only demicontractive. Therefore, the RED-PRO framework provides an optimization interpretation to PnP-PGD above along with convergence results when the conditions of \cref{theo:constant} hold. However, \cref{eq:REDPRO} as well as PnP-PGD may converge even when condition (H3) is not satisfied. In this case, \cref{algo:REDPROHSDM} converges to a solution of a special case of the RED-PRO formulation, as stated by the next theorem.    

\begin{theorem}
Let $f(\cdot)$ be a continuous $d$-demicontractive denoiser and $\ell(\cdot;\textbf{y})$ is a proper convex l.s.c differentiable function whose gradient $\nabla\ell(\cdot;\textbf{y})$ is L-Lipschitz. Assume the following holds
\begin{enumerate}[leftmargin=1cm]
    \item[(W1)] $\alpha\in(0,\frac{1-d}{2})$.
    \item[(W2)] $\mu_k\equiv\mu\in (0,\,\frac{2}{L})$.
    \item[(W3)] $\fix(f)\cap\fix(G_\ell)=\emptyset$ and $\fix(T)\neq\emptyset$ where $T(\textbf{x})\triangleq f_\alpha\Big(\textbf{x}-\mu\nabla\ell (\textbf{x})\Big)$.
\end{enumerate}
Then, the sequence $\{\textbf{x}_k\}_{k\in\bb{N}}$ generated by \cref{algo:REDPROHSDM} converges an optimal solution of the following convex feasibility problem:
\begin{align}
\begin{split}
     \underset{\textbf{x}\in\mathbb{R}^n}{\min}\;\ell(\textbf{x})\equiv0 \quad \text{ s.t. }\textbf{x}=f_\alpha\Big(x-\mu\nabla\ell(\textbf{x})\Big),
\end{split}
\label{eq:feasible}
\end{align}
i.e., find $\textbf{x}\in\mathbb{R}^n$ such that $\textbf{x}=T(\textbf{x})$.
.
\end{theorem}

\begin{proof}
Under the assumption that $\fix(T)\neq\emptyset$, the mapping $T(\cdot)$ is an averaged quasi-nonexpansive operator as composition of averaged operators. Hence, $\fix(T)$ is closed and convex and problem \cref{eq:feasible} can be seen a special of the RED-PRO formulation with respect to the composed denoiser $T(\cdot)$. Thus, the conditions of \cref{theo:diminishing} and \cref{theo:constant} hold and the iteration $\textbf{x}_{k+1}=T(\textbf{x}_k)=f_\alpha\Big(\textbf{x}_k-\mu\nabla\ell(\textbf{x}_k)\Big)$ converges to an arbitrary fixed-point of $T$.  
\end{proof}

Several remarks are to be made with regard to the above theoretical results. As discussed earlier, PnP-PGD can be seen a special case of \cref{algo:PGM} when the denoiser is strongly quasi-nonexpansive, hence, the above theorem complements \cref{theo:constant} with respect to the convergence of PnP-PGD. However, notice that unlike (H3), condition (W3) depends on $\mu$. Consequently, the solution of (\ref{eq:feasible}) is sensitive to the choice of $\mu$, which is consistent with previous observations in the literature on the convergence of PnP \cite{wei2020tuning}. While the convergence of PnP-PGD (and other variants) to a fixed-point was studied and proven before, here, we first present it a solution of a convex optimization problem, establishing the connection between the PnP and RED-PRO frameworks. Furthermore, it has been shown \cite{meinhardt2017learning,sun2019online,ryu2019plug} that other PnP variants, e.g., PnP-ADMM and PnP primal-dual hybrid gradient
method (PnP-PDHG), satisfy the same fixed-point equation as PnP-PGM 
\begin{equation}
    \textbf{x}^\ast=f\Big(\textbf{x}^\ast-\mu\nabla\ell(\textbf{x}^\ast;\textbf{y})\Big).
    \label{eq:fixedpointequation}
\end{equation} 
Therefore, under the conditions stated above, any algorithm that converges to a fixed-point satisfying \eqref{eq:fixedpointequation}, leads in principle to a solution of the RED-PRO formulation. Examining the latter, we observe that in general the solution of \cref{eq:feasible} is not unique and \cref{algo:REDPROHSDM} converges to an arbitrary fixed-point. This may explain why while different versions of PnP share the same fixed-point equation \cref{eq:fixedpointequation}, they converge to different solutions. Thus, we can gain more control on the obtained solution by formulating the following problem
\begin{align}
\begin{split}
     \underset{\textbf{x}\in\mathbb{R}^n}{\min}\;\frac{1}{2}\norm{\textbf{x}-\textbf{u}}^2 \quad \text{ s.t. }\textbf{x}=T(\textbf{x})\triangleq f_\alpha\Big(x-\mu\nabla\ell(\textbf{x})\Big),
\end{split}
\label{eq:feasibleproj}
\end{align}
for some desired $\textbf{u}\in\mathbb{R}^n$, which can be solved by  \cref{algo:REDPROHSDM} with a diminishing step size, leading to $P_{\fix(T)}(\textbf{u})$. This allows, for example, to obtain the minimal norm solution that satisfies \cref{eq:fixedpointequation} by setting $\textbf{u}=0$.

\subsection{Relaxed RED-PRO}
So far, we have introduced the RED-PRO framework which utilizes the fixed-point set of a denoiser as a regularization, and we have shown its connection to both PnP and RED. However, as mentioned earlier, the fixed-point sets of practical denoisers might be narrow, leading to unfavorable recovery solutions. To circumvent this limitation, we here relax the hard constraint of \cref{eq:REDPRO} and replace it with a squared distance penalty, leading to a relaxed RED-PRO minimization problem
\begin{equation}
                \hat{\bf x}=\underset{\textbf{x}\in\bb{R}^n}{\arg\min}\;\, \ell(\textbf{x};\textbf{y})+\frac{\lambda}{2}\norm{\textbf{x}-P_{\fix(f)}(\textbf{x})}^2.
                \label{eq:RRPoptimization}
\end{equation}
As before, when the denoiser is demicontractive, the above problem is a well-defined convex minimization. Here $\lambda>0$ balances between the log-likelihood term and the regularization, allowing us to control the distance of the optimal solution from the fixed-point set. Notice that when $\lambda\rightarrow\infty$ (or it sufficiently large),  problem \eqref{eq:RRPoptimization} reduces to problem \eqref{eq:REDPRO}. 
% Therefore the optimization \eqref{eq:RRPoptimization} generalizes the RED-PRO approach.   

Denoting by $E(\textbf{x})$ the objective function of \eqref{eq:RRPoptimization}, we have that
\begin{equation}
    \nabla E(\textbf{x})=\nabla \ell(\textbf{x};\textbf{y})+\lambda\Big(\textbf{x}-P_{\fix(f)}(\textbf{x})\Big).
    \label{eq:RRPgradient}
\end{equation}
Note that the latter expression is similar to the gradient \cref{eq:REDgradient} of the original RED framework, where we replace the denoiser with the projection onto its fixed-point set. Therefore, the solution of \eqref{eq:RRPoptimization} can be found by various optimization solvers, such as SD and ADMM, assuming one can compute the projection onto the fixed-point set
\begin{equation}
     P_{\fix(f)}(\textbf{x})=\underset{\textbf{v}\in\bb{R}^n}{\arg\min}\;\, \frac{1}{2}\norm{\textbf{x}-\textbf{v}}^2\;\,
    s.t.\;\, \textbf{v}\in\fix(f).
\end{equation}
Notice that the above problem is a special case of problem \eqref{eq:REDPRO} where the objective is $\frac{1}{2}\norm{\textbf{x}-\textbf{v}}^2$ which satisfies the conditions of \cref{theo:diminishing}. Therefore, given a point $\textbf{x}_0$, we can compute the projection $P_{\fix(f)}(\textbf{x}_0)$ using the following iterative scheme
\begin{equation}
    \textbf{x}_{j+1}=f_\alpha\Big(t_j\textbf{x}_0+(1-t_j)\textbf{x}_j\Big),
\end{equation}
where $\alpha>0$ and the sequence $\{t_k\}_{k\in\bb{N}}$ meet the requirements of \cref{theo:diminishing}. As shown in \cref{app:Halpern}, the above update rule can be rewritten equivalently as
\begin{equation}
    \textbf{x}_{j+1}=t_j\textbf{x}_0+(1-t_j)f_\alpha(\textbf{x}_j),
    \label{eq:Halpern}
\end{equation}
which is the well-known Halpern iteration \cite{halpern1967fixed, chidume2010iterative, muarucster2011strong,lieder2017convergence}. 

The relaxed RED-PRO approach offers a broader search domain at the expense of higher computational load, since it require multiple activations of the denoising engine in each iteration. As typically done, we limit ourselves in practice to a small number of inner iterations to reduce complexity, yielding an approximation of the projection at each iteration. 

As a side note, we mention that for demicontractive denoiser it holds that $\langle \textbf{x}-f(\textbf{x}), \textbf{x}-P_{\fix(f)}\rangle\geq 0$ implying that $\textbf{g}=-\lambda\big(\textbf{x}-f(\textbf{x})\big)$ is a descent direction of the regularization term $\rho(\textbf{x})=\frac{\lambda}{2}\norm{\textbf{x}-P_{\fix(f)}(\textbf{x})}^2$. Hence, an interesting direction to reduce complexity is replacing the term $-\lambda\Big(\textbf{x}-P_{\fix(f)}(\textbf{x})\Big)$ in \cref{eq:RRPgradient} with $\textbf{g}$, bringing us back to the original RED algorithms. However, it remains unclear whether this indeed leads to a solution of \cref{eq:RRPoptimization}.

% \begin{algorithm}
% \caption{RED-PRO via SD}
% \begin{algorithmic}[1]
% \item\textbf{Input:} ${\bf x}_0\in\bb{R}^n$, $\{t_k\}_{k\in\bb{N}}$, $\alpha\in(0,1)$, $\lambda,\,\mu,\,J>0$, $f(\cdot)$.
% \item \textbf{for} $k=0,1,2,...$ \textbf{do}:
% \newline \hphantom\quad   \textbf{for} $j=0,1,2,...J$ \textbf{do}:
% \begin{equation*}
%     \textbf{x}_{k,j+1}=f_\alpha\Big(t_j\textbf{x}_k+(1-t_j)\textbf{x}_{k,j}\Big)
% \end{equation*}
% $\quad\textbf{x}_{k+1}=\textbf{x}_k-\mu\Big(\nabla\ell(\textbf{x}_k;y)+\lambda(\textbf{x}_k-\textbf{x}_{k,J+1})\Big)$.
% \item \textbf{Output:} $\textbf{x}_{k+1}$.
% \end{algorithmic}
% \label{algo:RRPviaSD}
% \end{algorithm}

%=======================================================

% \subsubsection{Approximate Fixed-Points}
Next, we further relax our problem by considering richer sets than the fixed-point one. To that end, we define for $\epsilon>0$ the $\epsilon$-approximate fixed-points as 
\begin{equation}
    \fix_\epsilon(f)\triangleq \{\textbf{x}\in\bb{R}^n\,:\;\norm{\textbf{x}-f(\textbf{x})}\leq\epsilon\}.
\end{equation}
Note that $\fix_\epsilon(f)$ can be seen as the set of \textquotedblleft almost clean\textquotedblright\,  images, thus, it can act as a valid and wide domain for searching restored solutions. A possible approach to exploit $\fix_\epsilon(f)$ is defining an adaptive averaged denoiser  
\begin{equation}
    f_\epsilon({\textbf{x}})\triangleq \alpha_\epsilon(\textbf{x})\textbf{x}+\Big(1-\alpha_\epsilon(\textbf{x})\Big)f(\textbf{x}), 
    \label{eq:modifieddenoiser}
\end{equation}
where $\alpha_\epsilon(x)=\frac{\epsilon}{\max(\epsilon,\norm{\textbf{x}-f(\textbf{x})})}$. Now notice that $\fix(f_\epsilon)=\fix_\epsilon(f)$, hence, the modified denoiser \eqref{eq:modifieddenoiser} can be utilized via the algorithms derived earlier to find inverse solutions constrained by $\fix_\epsilon(f)$.
While empirically this may lead to satisfactory results, the approximate fixed-point set is generally not convex for demicontractive or even nonexpansive functions. Consequently, the projection onto $\fix_\epsilon(f)$ is not well-defined and no convergence guarantees can be given. One option standing ahead of us is to restrict the valid denoisers to those for which this set is necessarily convex. 

Another option for circumventing the above difficulty is to consider a dilated fixed-point set, defined for some $\delta>0$ as
\begin{equation}
    B_\delta(f)\triangleq\{\textbf{x}\in\bb{R}^n\;:\; \norm{\textbf{x}-P_{\fix(f)}(\textbf{x})}\leq \delta\}. 
\end{equation}
The set $B_\delta(f)$ is closed and convex for any demicontractive denoiser. Moreover, as illustrated in Fig.\,\ref{fig:sethierarchy}, for an appropriate choice of $\delta$ we have that $B_\delta(f)\subseteq \fix_\epsilon(f)$ as stated next.

\begin{theorem}
Let $f(\cdot)$ be $d$-demicontractive function and consider $\delta\in[0,\alpha\epsilon]$ for some $\epsilon>0$ and $\alpha\in(0,\frac{1-d}{2}]$. Then, $    B_\delta(f)\subseteq \fix_\epsilon(f)$.
\label{theo:dilation}
\end{theorem}
\begin{proof}
See \cref{app:dilation}.
\end{proof}

\begin{figure}
    \centering
    \includegraphics[trim={1cm 0cm 1cm 0cm},clip,height = 6cm, width = 0.7\linewidth]{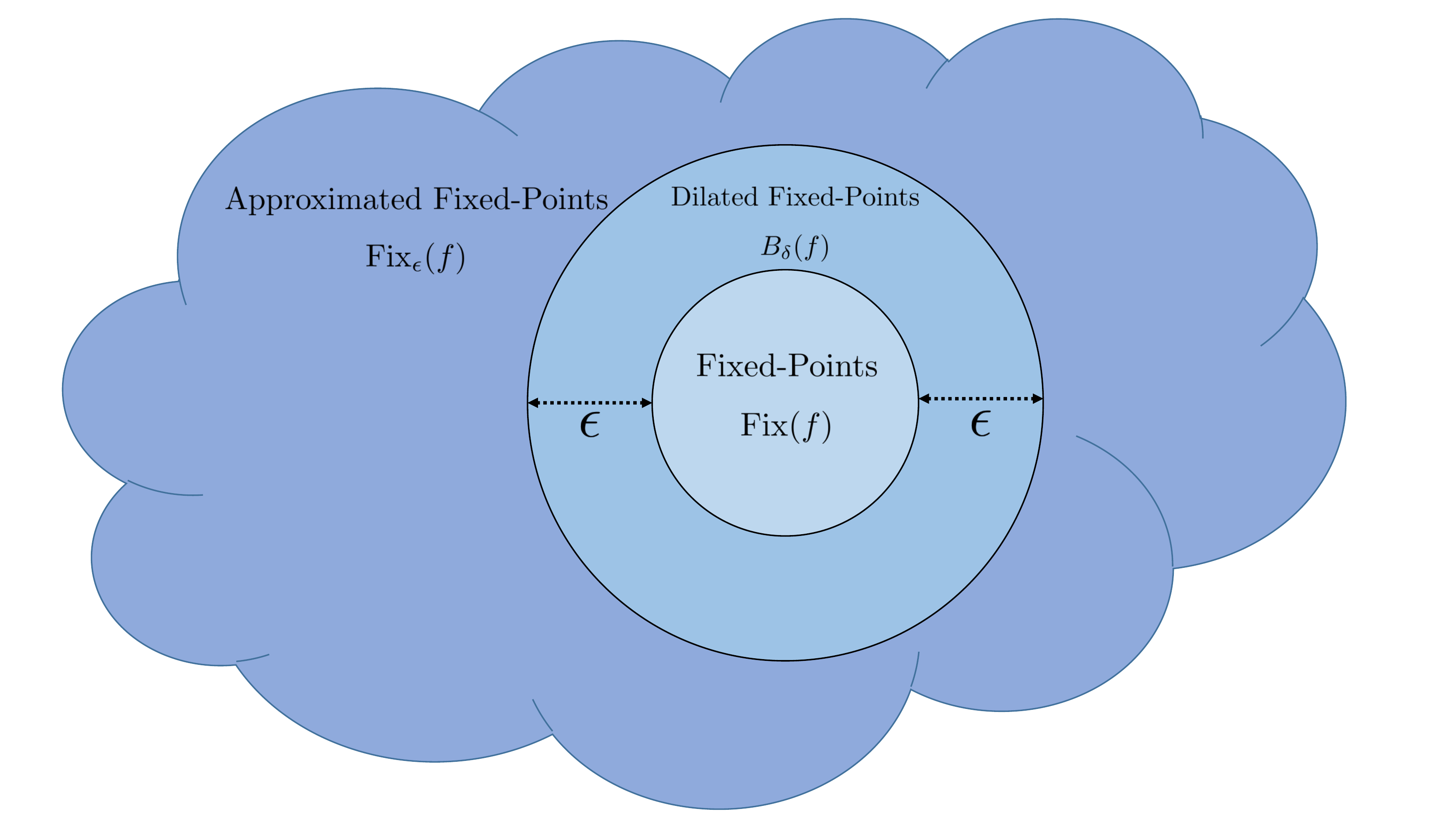}
    \caption{Hierarchy Illustration of fixed-point related sets. }
    \label{fig:sethierarchy}
\end{figure}

\noindent Exploiting the above result, we formulate the relaxed RED-PRO (RRP) problem
\begin{equation}
                \hat{\bf x}_\text{RRP}=\underset{\textbf{x}\in\bb{R}^n}{\arg\min}\;\, \ell(\textbf{x};\textbf{y})+\frac{\lambda}{2}\norm{\textbf{x}-P_{B_\delta(f)}(\textbf{x})}^2,
                \label{eq:dialtedRRP}
\end{equation}
where $P_{B_\delta(f)}$ is the projection onto $B_\delta(f)$. The gradient of the above objective is
\begin{equation}
    \nabla E_\text{RRP}=\nabla \ell(\textbf{x};\textbf{y})+\lambda\Big(\textbf{x}-P_{B_\delta(f)}(\textbf{x})\Big),
\end{equation}
requiring to compute the projection $P_{B_\delta(f)}$. Fortunately, for any $\textbf{x}\notin B_\delta(f)$, the projection is given by \cite{bargetz2018convergence}
\begin{equation}
    P_{B_\delta(f)}(\textbf{x})=  
    \frac{\delta}{\norm{\textbf{x}-P_{\fix(f)}(\textbf{x})}}\textbf{x} +\bigg(1-\frac{\delta}{\norm{\textbf{x}-P_{\fix(f)}(\textbf{x})}}\bigg)P_{\fix(f)}(\textbf{x}).
\end{equation}
\noindent Substituting the above expression into \cref{eq:dialtedRRP}, we obtain
\begin{equation}
                \hat{\bf x}_\text{RRP}=\underset{\textbf{x}\in\bb{R}^n}{\arg\min}\;\, \ell(\textbf{x};\textbf{y})+\frac{\lambda}{2}\Big(\norm{\textbf{x}-P_{\fix(f)}(\textbf{x})}-\delta\Big)^2.
                \label{eq:deltaRRP}
\end{equation}
The latter minimization problem consists of another  design parameter $\delta$ which allows to control the distance of the solution from the fixed-point set, thus, increasing our search domain. When $\delta$ is set to zero, we return to problem \eqref{eq:RRPoptimization}. In \cref{algo:RRPviaSD} we detail the overall solution of the relaxed RED-PRO with parameter $\delta$ using steepest descent as an example.

\begin{algorithm}
\caption{Relaxed RED-PRO via SD}
\begin{algorithmic}[1]
\item\textbf{Input:} ${\bf x}_0\in\bb{R}^n$, $\{t_k\}_{k\in\bb{N}}$, $\alpha\in(0,1)$, $\lambda,\,\mu,\,J,\,N\,\delta>0$ and $f(\cdot)$.
\item \textbf{for} $k=0,1,2,...,N-1$ \textbf{do}:
% \newline \hphantom\quad
\begin{itemize}
    \item[] $\textbf{x}_{k,0}=\textbf{x}_{k}$.
    \item[] \textbf{for} $j=0,1,2,...J-1$ \textbf{do}:
\begin{equation*}
    \textbf{x}_{k,j+1}=f_\alpha\Big(t_j\textbf{x}_k+(1-t_j)\textbf{x}_{k,j}\Big)
\end{equation*}
\item[] $\textbf{v}_k=\frac{\delta}{\norm{\textbf{x}_k-\textbf{x}_{k,J}}}\textbf{x}_k+\Big(1-\frac{\delta}{\norm{\textbf{x}_k-\textbf{x}_{k,J}}}\Big)\textbf{x}_{k,J}$.\newline
\item[] $\textbf{x}_{k+1}=\textbf{x}_k-\mu\Big(\nabla\ell(\textbf{x}_k;y)+\lambda(\textbf{x}_k-\textbf{v}_k)\Big).$
\end{itemize}
\item \textbf{Output:} $\textbf{x}_{k+1}$.
\end{algorithmic}
\label{algo:RRPviaSD}
\end{algorithm}

To summarize this part, we described here relaxed versions of the RED-PRO framework where we use the distance function as regularization and we extend our search to certain approximate fixed-points. This allows us to obtain stable inverse solutions even when the denoiser exhibits a limited fixed-point set.

% \subsubsection{Enlarged Fixed-Point Set}

% \subsection{RED-Cut}

%===========================================================
%===========================================================

\section{Experiments}
\label{sec:results}

In this section we evaluate the performance of the proposed RED-PRO framework. AS the goal of this study is to bridge between RED and PnP and to provide theoretical justifications rather than achieving state-of-the-art results, we compare ourselves to the original RED approach alone, for the tasks of image deblurring and super-resolution. We follow the same line of experiments performed in \cite{romano2017little}, where we employ NLM and TNRD denoisers, rather than the median filter and TNRD. We use NLM to exemplify the application of the RED-PRO with denoisers that may be expansive for certain parameters. The TNRD denoiser, pretrained to remove WGN whose standard deviation equals 5,\footnote{We use TNRD to address any arbitrary noise level $\sigma$ by exploiting the relation $f_\sigma(\textbf{x})=\frac{\sigma}{5}f_5(\frac{5}{\sigma}\textbf{x}).$} is brought to demonstrate the full extent of the proposed framework. 

For fair comparison, the parameters that are common to both RED and RED-PRO are set according to the values given in \cite{romano2017little}.\footnote{For RED algorithms we use the code published by the authors of \cite{romano2017little} at \url{https://github.com/google/RED}.} 
Additional parameters of introduced by the RED-PRO framework are tuned manually to achieve higher peak signal to noise ratio (PSNR).

\subsection{Convergence}
As a start, we demonstrate the convergence of RED-PRO with the HSD method for various denoisers. To that end, we consider the task of uniform deblurring whose complete setup is detailed in the following subsection, with the exception of the relaxation parameter that we set here to $\alpha=0.5$. 

We show the evolution of the relative error the fidelity term throughout the iterations in Fig.\,\ref{fig:convergence}(a) we, while Fig.\,\ref{fig:convergence}(b) displays a measure for the solution proximity to the fixed-point set over iterations. We observe that we obtain small relative errors in both measures, where rapid decline the fidelity term is seen for all denoisers, converging to different values as expected. In addition, for all of the denoisers, the solution converge to a fixed point, however, the rate of convergence varies for different denoiser where BM3D exhibits notably greater decrease than other denoisers.\footnote{Due to this dramatic decline of BM3D, it may seem that the errors of other denoisers are saturated. However, this is not the case, as their errors gradually decrease.} We point out that there is a tradeoff between the rate of convergence of the fidelity term and that of the fixed-point errors, since in general large step sizes lead to large decrease in the fidelity term while decaying step sizes are required to ensue convergence to a fixed-point. Finally, we note that while these measures indeed demonstrate the convergence of the proposed approach, they might not accurately indicate the quality of the resultant images.   

\begin{figure}[h]
    \centering
    \includegraphics[trim={0cm 3.2cm 0cm 3cm},clip,height = 7cm, width = 0.9\linewidth]{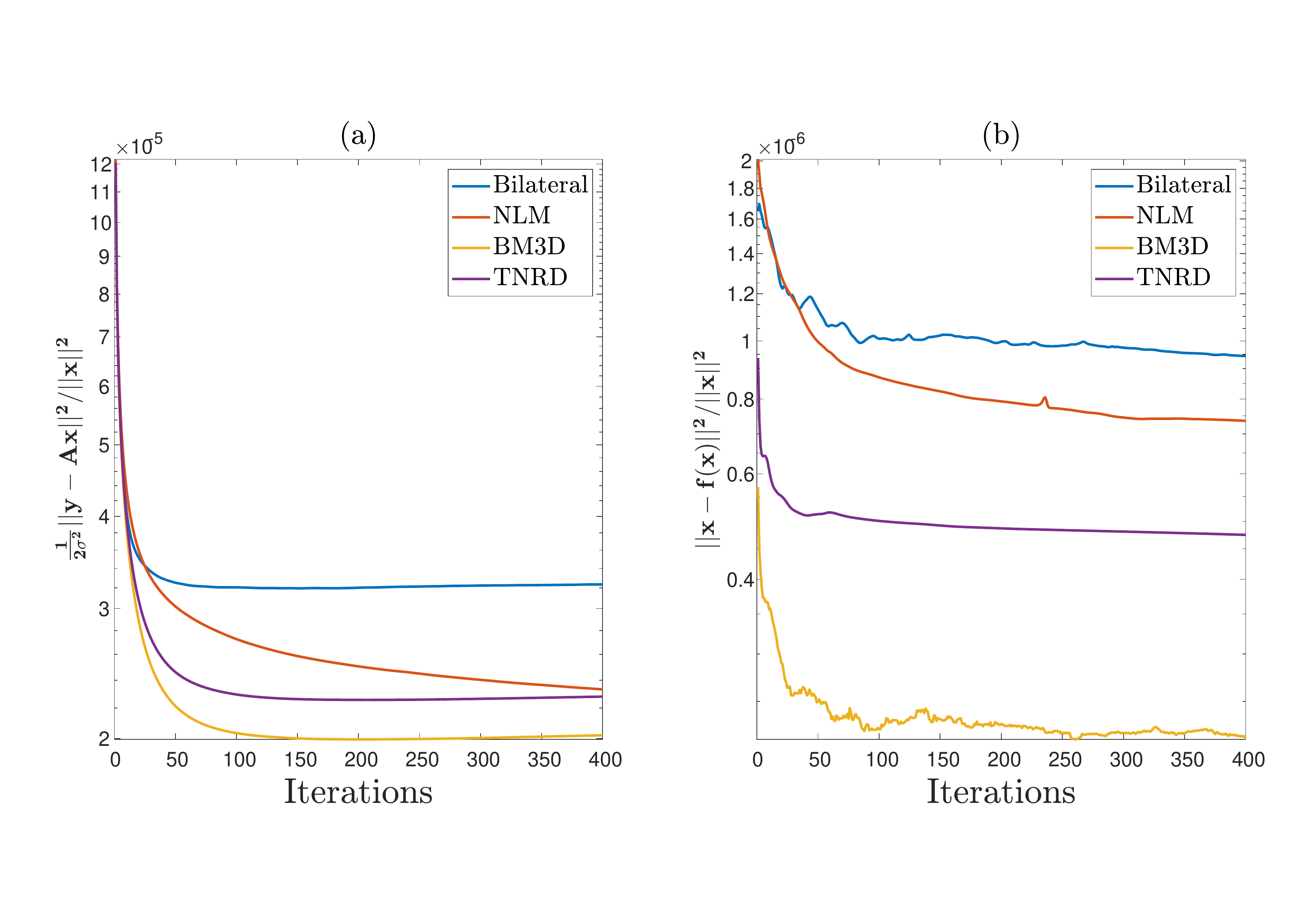}
    \caption{An illustration of the convergence of RED-PRO via HSD using various denoisers. Here we degrade the image \textbf{Starfish}, degraded by a uniform PSF and restore it by \cref{algo:REDPROHSDM} with diminishing step size and $\alpha=0.5$. (a) The evolution of the fidelity term during the iterations. (b) A measure of the solution proximity to the fixed-point set throughout the iterations.}
    \label{fig:convergence}
\end{figure}

\subsection{Image Deblurring}
As done in \cite{romano2017little}, we carry out here the synthetic non-blind deblurring experiments first described in \cite{dong2012nonlocally}. To that end, we perform the following process on the RGB images provided by the authors of \cite{romano2017little}. Each image is convolved either with a $9\times9$ uniform point spread function (PSF) or a 2D Gaussian function with a standard deviation of 1.6. Then, we contaminate the resulted blurred images with an additive WGN with $\sigma=\sqrt{2}$. The inverse operation is performed by first converting the RGB image to YCbCr color-space, applying the deblurring technique only on the luminance channel and then converting the result back to RGB to obtain the final image. The complete details of the parameter values for the uniform deblurring case are given in \cref{table:uniformparams}. For Gaussian deblurring we use the same values excluding $\sigma_f$ and $\lambda$ which we set to 4.1 and 0.01 respectively.
Throughout all the experiments, we use a constant step size $\mu_0=\frac{2}{\sigma_f^{-2}+\lambda}$ for the steepest descent methods, while for HSD we use a a diminishing step size $\mu_k=\mu_0k^{-0.1}$.

\cref{table:NLM-UB} and \cref{table:NLM-GB} detail the recovery results of RED and RED-PRO algorithms when NLM is used, while \cref{table:TNRD-UB} and \cref{table:TNRD-GB} provides the results when TNRD is employed. We evaluate the performance of the different techniques using PSNR measure, where higher is better, computed on the luminance channel of the ground truth and the restored images. We note that we show the results only for a chosen subset of images, while the displayed average value is computed over the entire set of images.

As can be seen, when NLM is used within the discussed frameworks, the RED-PRO approach displays an apparent improvement, in particular for the case of uniform deblurring. When TNRD is incorporated, 
the relaxed RED-PRO framework still leads to enhanced results in comparison to RED, but the  performance gap is small. Furthermore, we notice that SD-based methods require more iterations which is consistent with the observations given in \cite{romano2017little}. 
We present visual evidence for the performance of RED-PRO in Fig.\,\ref{fig:uniformblur} and Fig.\,\ref{fig:gaussianblur} for uniform and Gaussian blur kernels, respectively. These images support the previous results where RED-PRO offer an improvement over the RED framework when NLM is utilized. In the other case, both approaches achieve similar performance with a clear advantage for applying TNRD rather than employing NLM.

\begin{table*}
    \centering
\scriptsize
\renewcommand{\arraystretch}{1.2}
\renewcommand{\tabcolsep}{8.2pt}
\begin{tabular}{|c|c|c|c|c|c|c|c|c|c|c|}
\hline 
\multirow{2}{*}{Parameter} & \multicolumn{3}{c|}{RED} & \multicolumn{4}{c|}{RED-PRO} & \multicolumn{3}{c|}{Approximate RED-PRO}\tabularnewline
\cline{2-11} \cline{3-11} \cline{4-11} \cline{5-11} \cline{6-11} \cline{7-11} \cline{8-11} \cline{9-11} \cline{10-11} \cline{11-11} 
 & FP & ADMM & SD & HSD & FP & ADMM & SD & FP & ADMM & SD\tabularnewline
\hline 
$N$ & 200 & 200 & 1500 & 400 & 200 & 200 & 1500 & 200 & 200 & 1500\tabularnewline
\hline 
$\sigma_{f}$ & 3.25 & 3.25 & 3.25 & 3.25 & 3.25 & 3.25 & 3.25 & 3.25 & 3.25 & 3.25\tabularnewline
\hline 
$\lambda$ & 0.02 & 0.02 & 0.02 & --- & 0.02 & 0.02 & 0.02 & 0.02 & 0.02 & 0.02\tabularnewline
\hline 
$\beta$ & --- & 0.001 & --- & --- & --- & 0.001 & --- & --- & 0.001 & ---\tabularnewline
\hline 
$m_{1}$ & *CF & *CF & --- & --- & *CF & *CF & --- & *CF & *CF & ---\tabularnewline
\hline 
$m_{2}$ & --- & 1 & --- & --- & --- & 1 & --- & --- & 1 & ---\tabularnewline
\hline 
$\alpha$ & --- & --- & --- & 0.035 & 1 & 1 & 1 & 1 & 1 & 1\tabularnewline
\hline 
$J$ & --- & --- & --- & --- & 3 & 3 & 3 & 3 & 3 & 3\tabularnewline
\hline 
$\delta$ & --- & --- & --- & --- & --- & --- & --- & 0.0001 & 0.0001 & 0.0001\tabularnewline
\hline 
\end{tabular}
\newline {\raggedright *Closed-form using FFT \par}
\vspace{3pt}
    \caption{The set of parameters used in RED and RED-PRO frameworks for the task of deblurring images corrupted by a uniform PSF. The description of ADMM inner parameters $m_1$ and $m_2$ is given in \cite{romano2017little}.}
    \label{table:uniformparams}
\end{table*}

\begin{table*}
    \centering
\scriptsize
\renewcommand{\arraystretch}{1.3}
\begin{tabular}{|c|c|c|c|c|c|c|c|c|c|c|}
\hline 
\multirow{2}{*}{\backslashbox{\tiny Image}{\tiny Method}} & \multicolumn{3}{c|}{RED} & \multicolumn{1}{c|}{RED-PRO} & \multicolumn{3}{c|}{Relaxed RED-PRO $\delta=0$}& \multicolumn{3}{c|}{Relaxed RED-PRO $\delta>0$}
\tabularnewline
\cline{2-11} \cline{3-11} \cline{4-11} \cline{5-11} \cline{6-11} \cline{7-11} \cline{8-11} \cline{9-11} \cline{10-11} \cline{11-11} 
 & FP & ADMM & SD & HSD & FP & ADMM & SD & FP & ADMM & SD\tabularnewline
\hline 
Bike & 23.07 & 23.12 & 24.41 & \hlcell{24.95} & 23.09 & 23.16 & 24.46 & 23.08 & 23.13 & 24.51\tabularnewline
\hline 
Butterfly & 24.93 & 25.01 & 26.93 & \hlcell{27.47} &
24.93 & 25.00 & 26.96 & 24.83 & 24.87 & 26.89 \tabularnewline
\hline 
Flower & 26.21 & 26.25 & 27.89 & \hlcell{29.26} & 26.27 &
26.33 & 27.97 & 26.25 & 26.27 & 28.01
\tabularnewline
\hline 
Girl & 
31.29 &
31.32 &
31.86 &
\hlcell{32.68} &
31.35 &
31.42 &
31.93 &
31.40 &
31.43 &
31.98
\tabularnewline
\hline 
Hat & 
29.04&
29.06&
30.40&
\hlcell{31}.42&
29.15&
29.13&
30.49&
29.08&
29.05&
30.46
\tabularnewline
\hline 
% Leaves & 
% 24.40&
% 24.45&
% 26.47&
% \hlcell{27.07}&
% 24.36&
% 24.44&
% 26.44&
% 24.21&
% 24.29&
% 26.34
% \tabularnewline
% \hline 
% Parrots & 
% 28.82 &
% 28.92 &
% 30.06 &
% 30.03 &
% 28.82 &
% 28.90 &
% \hlcell{30.08} &
% 28.77 &
% 28.77 &
% 30.06 
% \tabularnewline
% \hline 
% Parthenon & 
% 25.14&
% 25.20&
% 26.98&
% \hlcell{29.35}&
% 25.22&
% 25.27&
% 27.10&
% 25.18&
% 25.26&
% 27.14
% \tabularnewline
% \hline 
% Plants & 
% 30.52&
% 30.47& 
% 32.14&
% \hlcell{33.55}&
% 30.55&
% 30.56&
% 32.20&
% 30.48&
% 30.53&
% 32.22
% \tabularnewline
% \hline 
% Raccoon & 
% 26.86&
% 26.83&
% 27.60&
% \hlcell{28.62}&
% 26.98&
% 26.98&
% 27.75&
% 26.90&
% 26.87&
% 27.76
% \tabularnewline
% \hline
% Starfish & 
% 24.98&
% 25.03&
% 26.99&
% \hlcell{28.96}&
% 25.02&
% 25.05&
% 27.08&
% 24.95&
% 24.99&
% 27.06
% \tabularnewline
% \hline
Average & 
26.84&
26.88&
28.34&
\hlcell{29.40}&
26.89&
26.93&
28.40&
26.83&
26.86&
28.40
\tabularnewline
\hline 
\end{tabular}
\vspace{5pt}
    \caption{Deblurring with NLM: Recovery results obtained by RED and RED-PRO evaluated on the set of images, provided by the authors of \cite{romano2017little}, which corrupted with uniform blur kernel. Performance is measured in PSNR [dB] where highest result are highlighted.}
    \label{table:NLM-UB}
\end{table*}

\begin{table*}
    \centering
\scriptsize
\renewcommand{\arraystretch}{1.3}
\begin{tabular}{|c|c|c|c|c|c|c|c|c|c|c|}
\hline 
\multirow{2}{*}{\backslashbox{\tiny Image}{\tiny Method}} & \multicolumn{3}{c|}{RED} & \multicolumn{1}{c|}{RED-PRO} & \multicolumn{3}{c|}{Relaxed RED-PRO $\delta=0$}& \multicolumn{3}{c|}{Relaxed RED-PRO $\delta>0$}
\tabularnewline
\cline{2-11} \cline{3-11} \cline{4-11} \cline{5-11} \cline{6-11} \cline{7-11} \cline{8-11} \cline{9-11} \cline{10-11} \cline{11-11} 
 & FP & ADMM & SD & HSD & FP & ADMM & SD & FP & ADMM & SD\tabularnewline
\hline 
Bike & 
26.10&
26.09&
25.95&
24.95&
26.21&
26.20&
26.05&
\hlcell{26.48}&
26.47&
26.28
\tabularnewline
\hline 
Butterfly & 
30.41&
30.40&
30.20&
27.24&
30.48&
30.47&
30.24&
\hlcell{30.64}&
30.63&
30.35
\tabularnewline
\hline 
Flower & 
30.18&
30.18&
30.13&
29.38&
30.26&
30.26&
30.21&
\hlcell{30.46}&
\hlcell{30.46}&
30.39
\tabularnewline
\hline 
Girl & 
33.11&
33.11&
33.10&
32.99&
33.14&
33.14&
33.13&
\hlcell{33.19}&
\hlcell{33.19}&
\hlcell{33.19}
\tabularnewline
\hline 
Hat & 
32.16&
32.16&
32.15&
31.55&
32.17&
32.18&
32.17&
\hlcell{32.25}&
\hlcell{32.25}&
32.24
\tabularnewline
\hline 
% Leaves & 
% 30.13&
% 30.12&
% 29.72&
% 26.92&
% 30.20&
% 30.20&
% 29.77&
% \hlcell{30.39}&
% 30.38&
% 29.88
% \tabularnewline
% \hline 
% Parrots & 
% 31.83&
% 31.86&
% 31.62&
% 30.12&
% 31.92&
% 31.93&
% 31.68&
% \hlcell{32.13}&
% 32.12&
% 31.82
% \tabularnewline
% \hline 
% Parthenon & 
% 29.85&
% 29.85&
% 29.85&
% 29.82&
% 29.97&
% 29.97&
% 29.96&
% \hlcell{30.23}&
% \hlcell{30.23}&
% 30.22
% \tabularnewline
% \hline 
% Plants & 
% 34.25&
% 34.25&
% 34.19&
% 33.59&
% 34.27&
% 34.27&
% 34.21&
% \hlcell{34.40}&
% 34.39&
% 34.33
% \tabularnewline
% \hline 
% Raccoon & 
% 28.83&
% 28.84&
% 28.81&
% 28.99&
% 28.94&
% 28.94&
% 28.91&
% \hlcell{29.12}&
% \hlcell{29.12}&
% 29.08
% \tabularnewline
% \hline
% Starfish & 
% 30.57&
% 30.57&
% 30.47&
% 29.41&
% 30.65&
% 30.65&
% 30.55&
% \hlcell{30.84}&
% \hlcell{30.84}&
% 30.72
% \tabularnewline
% \hline
Average & 
30.67&
30.67&
30.56&
29.54&
30.75&
30.75&
30.63&
\hlcell{30.92}&
\hlcell{30.92}&
30.77
\tabularnewline
\hline 
\end{tabular}
\vspace{5pt}
    \caption{Deblurring with TNRD: Recovery results obtained by RED and RED-PRO evaluated on the set of images, provided by the authors of \cite{romano2017little}, which corrupted with uniform blur kernel. Performance is measured in PSNR [dB] where highest result are highlighted.}
    \label{table:TNRD-UB}
\end{table*}

\begin{figure*}[ht]
 \centering
 \includegraphics[trim={0cm 0.5cm 0cm 1cm},clip,height = 8cm, width = 0.8\linewidth]{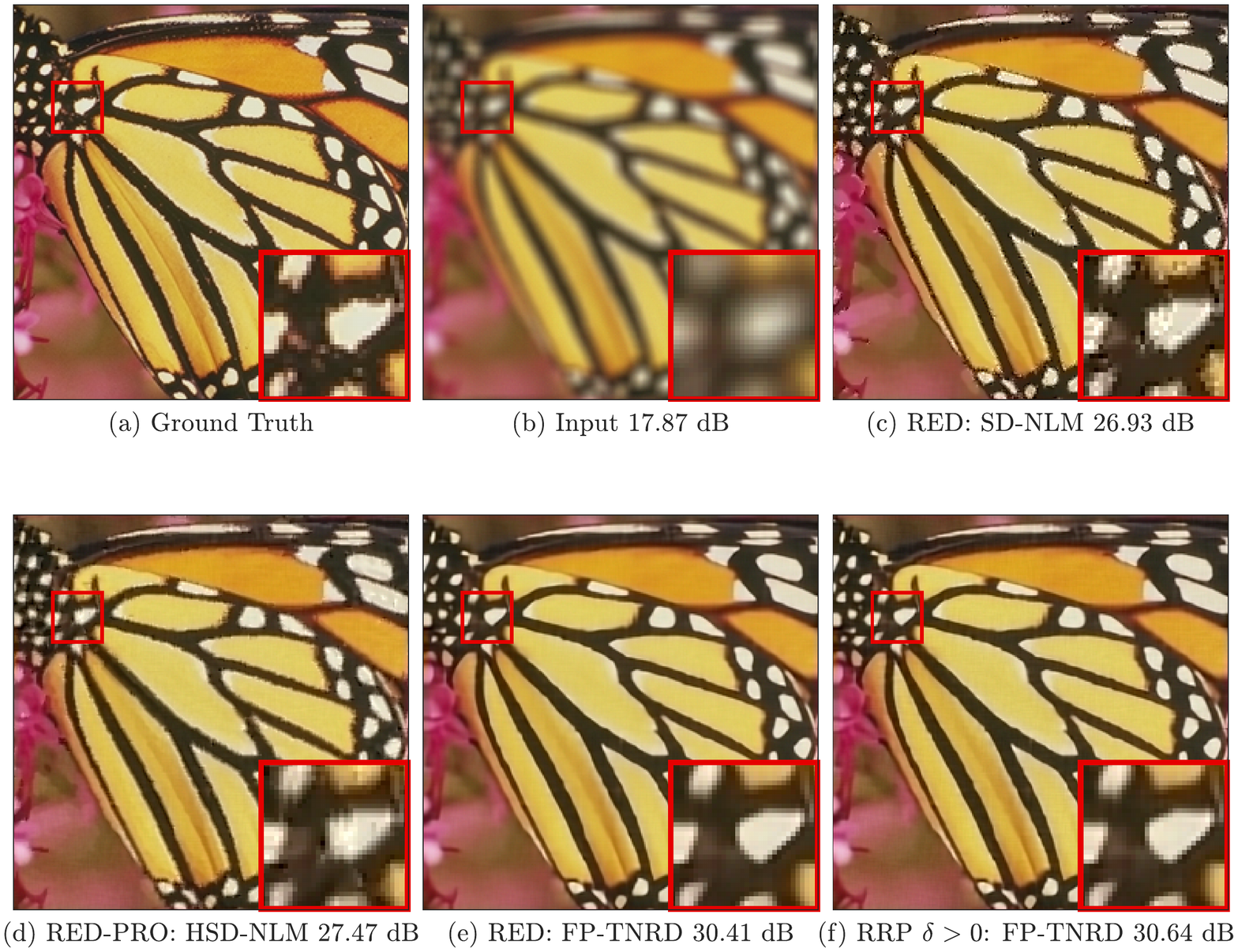}
\vspace{-5pt}
\caption{Visual comparison of deblurring the image \textbf{Butterfly}, degraded by a uniform
PSF, along with the corresponding PSNR [dB] score.}
  \label{fig:uniformblur}
 \end{figure*}

 \begin{table*}
    \centering
\scriptsize
\renewcommand{\arraystretch}{1.3}
\begin{tabular}{|c|c|c|c|c|c|c|c|c|c|c|}
\hline 
\multirow{2}{*}{\backslashbox{\tiny Image}{\tiny Method}} & \multicolumn{3}{c|}{RED} & \multicolumn{1}{c|}{RED-PRO} & \multicolumn{3}{c|}{Relaxed RED-PRO $\delta=0$}& \multicolumn{3}{c|}{Relaxed RED-PRO $\delta>0$}
\tabularnewline
\cline{2-11} \cline{3-11} \cline{4-11} \cline{5-11} \cline{6-11} \cline{7-11} \cline{8-11} \cline{9-11} \cline{10-11} \cline{11-11} 
 & FP & ADMM & SD & HSD & FP & ADMM & SD & FP & ADMM & SD\tabularnewline
\hline 
Bike & 
25.30&
25.45&
26.97&
\hlcell{27.13}&
25.38&
25.54&
27.03&
25.42&
25.57&
27.08
\tabularnewline
\hline 
Butterfly & 
27.14&
27.28&
29.61&
\hlcell{29.87}&
27.18&
27.36&
29.64&
27.21&
27.37&
29.66
\tabularnewline
\hline 
Flower & 
28.47&
28.62&
30.75&
\hlcell{31.64}&
28.58&
28.73&
30.85&
28.61&
28.77&
30.92
\tabularnewline
\hline 
Girl & 
33.32&
\hlcell{33.45}&
33.33&
32.97&
33.33&
33.41&
33.29&
33.29&
33.36&
33.27
\tabularnewline
\hline 
Hat & 
30.84&
30.95&
32.39&
\hlcell{32.74}&
30.92&
31.04&
32.45&
30.96&
31.08&
32.47
\tabularnewline
\hline 
% Leaves & 
% 27.50&
% 27.66&
% 29.98&
% \hlcell{30.26}&
% 27.53&
% 27.70&
% 30.00&
% 27.51&
% 27.68&
% 30.00
% \tabularnewline
% \hline 
% Parrots & 
% 30.76&
% 30.89&
% 32.30&
% \hlcell{32.34}&
% 30.82&
% 30.97&
% 32.30&
% 30.81&
% 30.96&
% 32.30
% \tabularnewline
% \hline 
% Parthenon & 
% 26.93&
% 27.06&
% 28.45&
% \hlcell{28.86}&
% 27.05&
% 27.17&
% 28.54&
% 27.10&
% 27.22&
% 28.59
% \tabularnewline
% \hline 
% Plants & 
% 33.30&
% 33.50&
% 35.01&
% \hlcell{35.47}&
% 33.36&
% 33.56&
% 34.98&
% 33.39&
% 33.55&
% 34.97
% \tabularnewline
% \hline 
% Raccoon & 
% 29.16&
% 29.22&
% 30.20&
% \hlcell{30.52}&
% 29.35&
% 29.40&
% 30.34&
% 29.33&
% 29.41&
% 30.39
% \tabularnewline
% \hline
% Starfish & 
% 28.05&
% 28.19&
% 30.72&
% \hlcell{31.47}&
% 28.11&
% 28.28&
% 30.78&
% 28.13&
% 28.31&
% 30.82
% \tabularnewline
% \hline
Average & 
29.16&
29.30&
30.88&
\hlcell{31.21}&
29.24&
29.38&
30.93&
29.25&
29.39&
30.95
\tabularnewline
\hline 
\end{tabular}
\vspace{5pt}
    \caption{Deblurring with NLM : Recovery results obtained by RED and RED-PRO evaluated on the set of images, provided by the authors of \cite{romano2017little}, which corrupted with Gaussian blur kernel. Performance is measured in PSNR [dB] where highest result are highlighted.}
    \label{table:NLM-GB}
\end{table*}

 \begin{table*}
    \centering
\scriptsize
\renewcommand{\arraystretch}{1.3}
\begin{tabular}{|c|c|c|c|c|c|c|c|c|c|c|}
\hline 
\multirow{2}{*}{\backslashbox{\tiny Image}{\tiny Method}} & \multicolumn{3}{c|}{RED} & \multicolumn{1}{c|}{RED-PRO} & \multicolumn{3}{c|}{Relaxed RED-PRO $\delta=0$}& \multicolumn{3}{c|}{Relaxed RED-PRO $\delta>0$}\tabularnewline
\cline{2-11} \cline{3-11} \cline{4-11} \cline{5-11} \cline{6-11} \cline{7-11} \cline{8-11} \cline{9-11} \cline{10-11} \cline{11-11} 
 & FP & ADMM & SD & HSD & FP & ADMM & SD & FP & ADMM & SD\tabularnewline
\hline 
Bike & 
27.90&
27.90&
27.88&
27.36&
27.95&
27.94&
27.92&
\hlcell{28.02}&
\hlcell{28.02}&
28.00
\tabularnewline
\hline 
Butterfly & 
\hlcell{31.66}&
\hlcell{31.66}&
31.57&
30.55&
31.65&
31.64&
31.54&
\hlcell{31.66}&
31.65&
31.52
\tabularnewline
\hline 
Flower & 
32.05&
32.05&
32.05&
31.81&
32.05&
32.05&
32.04&
\hlcell{32.08}&
\hlcell{32.08}&
32.07
\tabularnewline
\hline 
Girl & 
\hlcell{34.44}&
\hlcell{34.44}&
34.10&
33.23&
34.40&
34.40&
34.00&
34.38&
34.38&
33.88
\tabularnewline
\hline 
Hat & 
33.29&
33.29&
\hlcell{33.30}&
33.07&
33.25&
33.25&
33.26&
33.26&
33.26&
33.27
\tabularnewline
\hline 
% Leaves & 
% 31.93&
% 31.93&
% \hlcell{31.95}&
% 30.79&
% 31.91&
% 31.91&
% 31.91&
% 31.92&
% 31.92&
% 31.87
% \tabularnewline
% \hline 
% Parrots & 
% \hlcell{33.33}&
% 33.32&
% 33.19&
% 32.59&
% 33.29&
% 33.29&
% 33.15&
% 33.31&
% 33.30&
% 33.15
% \tabularnewline
% \hline 
% Parthenon & 
% 29.15&
% 29.15&
% 29.14&
% 28.96&
% 29.17&
% 29.17&
% 29.16&
% \hlcell{29.23}&
% \hlcell{29.23}&
% 29.20
% \tabularnewline
% \hline 
% Plants & 
% 36.38&
% 36.38&
% \hlcell{36.39}&
% 35.69&
% 36.26&
% 36.27&
% 36.25&
% 36.22&
% 36.22&
% 36.16
% \tabularnewline
% \hline 
% Raccoon & 
% 31.07&
% 31.08&
% 31.03&
% 30.82&
% 31.13&
% 31.14&
% 31.08&
% \hlcell{31.20}&
% \hlcell{31.20}&
% 31.11
% \tabularnewline
% \hline
% Starfish & 
% 32.49&
% 32.49&
% 32.46&
% 31.92&
% 32.49&
% 32.49&
% 32.44&
% \hlcell{32.52}&
% \hlcell{32.52}&
% 32.45
% \tabularnewline
% \hline
Average & 
32.15&
32.15&
32.09&
31.53&
32.14&
32.14&
32.07&
\hlcell{32.16}&
\hlcell{32.16}&
32.06
\tabularnewline
\hline 
\end{tabular}
\vspace{5pt}
    \caption{Deblurring with TNRD: Recovery results obtained by RED and RED-PRO evaluated on the set of images, provided by the authors of \cite{romano2017little}, which corrupted with Gaussian blur kernel. Performance is measured in PSNR [dB] where highest result are highlighted.}
    \label{table:TNRD-GB}
\end{table*}

\begin{figure*}[ht]
 \centering
 \includegraphics[trim={0cm 0.5cm 0cm 1cm},clip,height = 8cm, width = 0.8\linewidth]{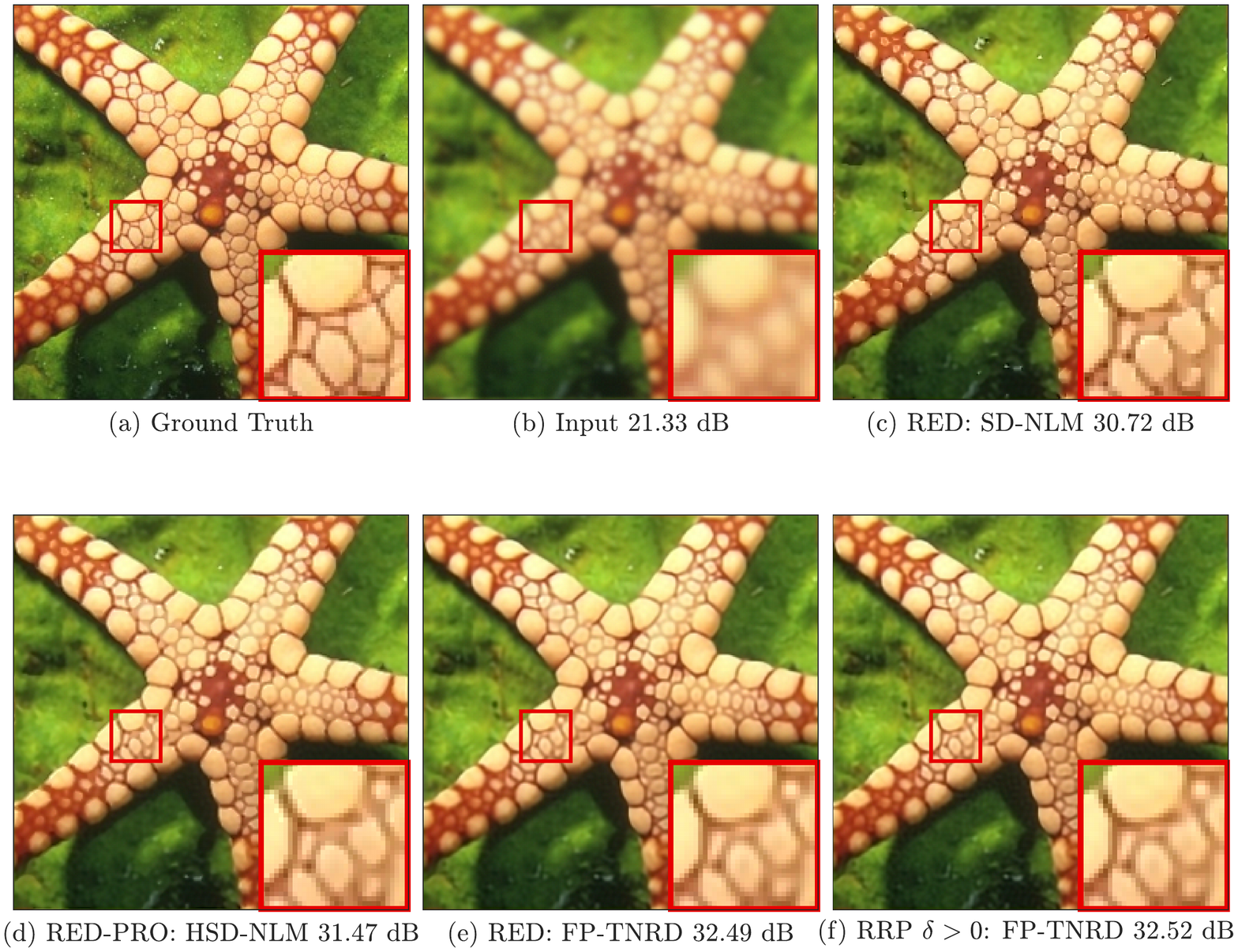}
\vspace{-5pt}
\caption{Visual comparison of deblurring the image \textbf{Starfish}, degraded by a Gaussian
PSF, along with the corresponding PSNR [dB] score.}
  \label{fig:gaussianblur}
 \end{figure*}

\subsection{Super-Resolution}
Similar to \cite{romano2017little}, we continue with super-resolution experiments where ground-truth high-resolution images are blurred with $7\times7$ Gaussian kernel with standard deviation of 1.6, downsampled by a factor of three at each axis and finally they are corrupted by an additive WGN with $\sigma=5$. To restore the images from their low-resolution versions, we convert them to YCbCr color-space and perform super-resolution where the chorma channels are up-scaled by bicubic interpolation and the luminance channel is processed by the discussed algorithms. The results are then transformed back to RGB, yielding the final super-resolved images. The parameter values are given by \cref{table:uniformparams} excluding $\sigma_f$, $\lambda$ and $m_1$ which are set to 3, 0.008 and 200 respectively.

We present the restoration results of the different variants of RED and RED-PRO frameworks when NLM and TNRD are used in \cref{table:NLM-SR} and \cref{table:TNRD-SR} respectively.
As observed, the RED-PRO approach leads to slightly better results than RED. This is visually supported by the recovered images shown in Fig.\,\ref{fig:superresolution} where we compare our approach to naive bicubic interpolation in addition to RED. Both quantitative and qualitative results point out the clear advantage of using TNRD over NLM.

 \begin{table*}
    \centering
\scriptsize
\renewcommand{\arraystretch}{1.3}
\begin{tabular}{|c|c|c|c|c|c|c|c|c|c|c|}
\hline 
\multirow{2}{*}{\backslashbox{\tiny Image}{\tiny Method}} & \multicolumn{3}{c|}{RED} & \multicolumn{1}{c|}{RED-PRO} & \multicolumn{3}{c|}{Relaxed RED-PRO $\delta=0$}& \multicolumn{3}{c|}{Relaxed RED-PRO $\delta>0$}\tabularnewline
\cline{2-11} \cline{3-11} \cline{4-11} \cline{5-11} \cline{6-11} \cline{7-11} \cline{8-11} \cline{9-11} \cline{10-11} \cline{11-11} 
 & FP & ADMM & SD & HSD & FP & ADMM & SD & FP & ADMM & SD\tabularnewline
\hline 
Bike & 
22.97&
23.04&
22.27&
23.10&
23.02&
23.08&
22.34&
23.05&
\hlcell{23.11}&
22.41
\tabularnewline
\hline 
Butterfly & 
25.27&
25.37&
24.23&
24.85&
25.30&
25.39&
24.29&
25.32&
\hlcell{25.40}&
24.37
\tabularnewline
\hline 
Flower & 
26.70&
26.82&
25.42&
\hlcell{27.21}&
26.79&
26.91&
25.48&
26.89&
27.02&
25.64
\tabularnewline
\hline 
Girl & 
30.75&
30.78&
30.03&
30.50&
30.84&
30.90&
30.13&
30.96&
\hlcell{31.02}&
30.19
\tabularnewline
\hline 
Hat & 
29.19&
29.27&
28.19&
28.73&
29.25&
29.34&
28.29&
29.28&
\hlcell{29.38}&
28.31
\tabularnewline
\hline 
% Leaves & 
% 24.60&
% 24.68&
% 23.77&
% 24.32&
% 24.64&
% 24.73&
% 23.79&
% 24.69&
% \hlcell{24.77}&
% 23.85
% \tabularnewline
% \hline 
% Parrots & 
% 28.35&
% 28.39&
% 27.84&
% 27.98&
% 28.42&
% 28.46&
% 27.87&
% 28.49&
% \hlcell{28.54}&
% 27.92
% \tabularnewline
% \hline 
% Parthenon & 
% 25.57&
% 25.64&
% 24.79&
% 25.71&
% 25.62&
% 25.71&
% 24.87&
% 25.73&
% \hlcell{25.78}&
% 24.97
% \tabularnewline
% \hline 
% Plants & 
% 30.38&
% 30.49&
% 28.96&
% 29.88&
% 30.47&
% 30.57&
% 28.95&
% 30.47&
% \hlcell{30.63}&
% 28.89
% \tabularnewline
% \hline 
% Raccoon & 
% 27.14&
% 27.20&
% 26.52&
% 27.30&
% 27.23&
% 27.29&
% 26.63&
% 27.28&
% \hlcell{27.33}&
% 26.62
% \tabularnewline
% \hline
% Starfish & 
% 25.90&
% 26.01&
% 24.88&
% \hlcell{26.69}&
% 25.98&
% 26.09&
% 24.98&
% 26.09&
% 26.19&
% 25.13
% \tabularnewline
% \hline
Average & 
26.98&
27.06&
26.08&
26.94&
27.05&
27.13&
26.15&
27.11&
\hlcell{27.20}&
26.21
\tabularnewline
\hline 
\end{tabular}
\vspace{5pt}
    \caption{Super-Resolution with NLM: Recovery results obtained by RED and RED-PRO evaluated on the set of images, provided by the authors of \cite{romano2017little}, which corrupted with Gaussian blur kernel and downsampled by a factor of 3 in each axis. Performance is measured in PSNR [dB] where highest result are highlighted.}
    \label{table:NLM-SR}
\end{table*}

 \begin{table*}
    \centering
\scriptsize
\renewcommand{\arraystretch}{1.3}
\begin{tabular}{|c|c|c|c|c|c|c|c|c|c|c|}
\hline 
\multirow{2}{*}{\backslashbox{\tiny Image}{\tiny Method}} & \multicolumn{3}{c|}{RED} & \multicolumn{1}{c|}{RED-PRO} & \multicolumn{3}{c|}{Relaxed RED-PRO $\delta=0$}& \multicolumn{3}{c|}{Relaxed RED-PRO $\delta>0$}\tabularnewline
\cline{2-11} \cline{3-11} \cline{4-11} \cline{5-11} \cline{6-11} \cline{7-11} \cline{8-11} \cline{9-11} \cline{10-11} \cline{11-11} 
 & FP & ADMM & SD & HSD & FP & ADMM & SD & FP & ADMM & SD\tabularnewline
\hline 
Bike & 
23.97&
23.96&
24.04&
23.22&
24.00&
23.98&
24.06&
24.03&
24.01&
\hlcell{24.09}
\tabularnewline
\hline 
Butterfly & 
27.26&
27.22&
\hlcell{27.37}&
24.95&
27.23&
27.19&
\hlcell{27.37}&
27.11&
27.06&
27.32
\tabularnewline
\hline 
Flower & 
28.24&
28.24&
28.23&
27.11&
28.26&
28.26&
28.25&
\hlcell{28.30}&
\hlcell{28.30}&
28.28
\tabularnewline
\hline 
Girl & 
32.08&
32.08&
32.08&
29.93&
\hlcell{32.09}&
\hlcell{32.09}&
32.08&
32.05&
32.06&
32.05
\tabularnewline
\hline 
Hat & 
30.35&
30.35&
30.36&
28.21&
\hlcell{30.37}&
30.36&
\hlcell{30.37}&
30.36&
30.35&
\hlcell{30.37}
\tabularnewline
\hline 
% Leaves & 
% 26.12&
% 26.10&
% 26.17&
% 24.47&
% 26.11&
% 26.09&
% 26.22&
% 26.06&
% 26.02&
% \hlcell{26.25}
% \tabularnewline
% \hline 
% Parrots & 
% 29.42&
% 29.41&
% 29.43&
% 27.72&
% 29.44&
% 29.44&
% \hlcell{29.50}&
% 29.44&
% 29.43&
% 29.46
% \tabularnewline
% \hline 
% Parthenon & 
% 26.52&
% 26.51&
% \hlcell{26.54}&
% 25.58&
% 26.52&
% 26.52&
% \hlcell{26.54}&
% 26.48&
% 26.48&
% 26.51
% \tabularnewline
% \hline 
% Plants & 
% 31.77&
% 31.77&
% 31.79&
% 29.44&
% 31.79&
% 31.79&
% \hlcell{31.82}&
% 31.76&
% 31.75&
% 31.77
% \tabularnewline
% \hline 
% Raccoon & 
% 27.97&
% 27.97&
% 27.98&
% 27.20&
% 28.01&
% 28.01&
% 28.01&
% 28.02&
% 28.02&
% \hlcell{28.03}
% \tabularnewline
% \hline
% Starfish & 
% 27.94&
% 27.93&
% 27.96&
% 26.83&
% 27.96&
% 27.95&
% 27.99&
% 27.98&
% 27.97&
% \hlcell{28.04}
% \tabularnewline
% \hline
Average & 
28.33&
28.32&
28.36&
26.79&
28.34&
28.33&
\hlcell{28.38}&
28.33&
28.31&
\hlcell{28.38}
\tabularnewline
\hline 
\end{tabular}
\vspace{5pt}
    \caption{Super-Resolution with TNRD Recovery results obtained by RED and RED-PRO evaluated on the set of images, provided by the authors of \cite{romano2017little}, which corrupted with Gaussian blur kernel and downsampled by a factor of 3 in each axis. Performance is measured in PSNR [dB] where highest result are highlighted.}
    \label{table:TNRD-SR}
\end{table*}

\begin{figure*}[ht]
 \centering
 \includegraphics[trim={0cm 0.5cm 0cm 1cm},clip,height = 8cm, width = 0.8\linewidth]{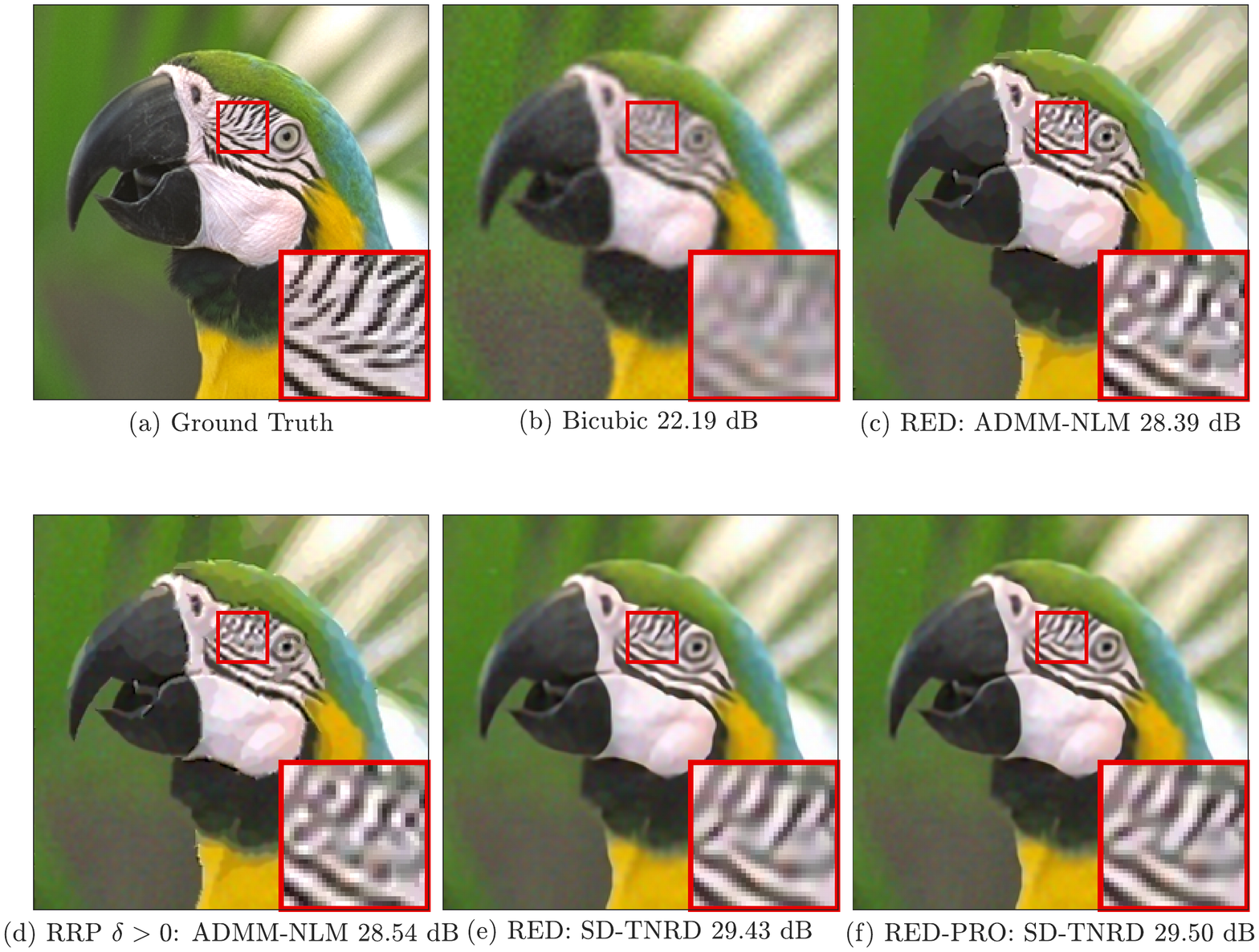}
\vspace{-5pt}
\caption{Visual comparison of super-resolutiin by a factor of 3 at each axis for the image \textbf{Parrots}, along with the corresponding PSNR [dB] score.}
  \label{fig:superresolution}
 \end{figure*}

\section{Conclusion}
\label{sec:conclusion}

The massive advancement in image denoising has led in recent years to the development of PnP and RED frameworks, which leverage the power of denoisers to solve other challenging inverse problems. These approaches have shown remarkable performance, leading to state-of-the-art results in unsupervised tasks such as image deblurring and super-resolution. However, PnP lacks an underlying global objective function, while RED demands restrictive conditions that are not met by many popular denoisers. 

In this paper, we address the above limitations. 
% As a start, we formulated an inverse problem where the Rockafellar is used as regularization. We show that when the denoiser residual is maximal cyclically monotone operator, then the RED algorithms iteratively solve our proposed problem. Thus, we provide global convergence guarantees of the RED framework for the case where the denoiser is non-differentiable. 
% In addition, we prove that a sufficient requirement is the denoiser to be cyclically firmly nonexpansive. However, a key question is whether this condition can be relaxed to other broader assumption.  
The main contribution of this work is the introduction of the RED-PRO framework, which relies on the fixed-point set of demicontractive denoisers. We formulate a general inverse problem that exploits the convexity of the fixed-point set and utilize it as a regularization. Then, we derive several iterative methods for solving the proposed problems and prove their convergence to global optimum. We discuss the relation of the presented framework to PnP and RED, including the connection of demicontractivity to popular assumptions. This shows that RED-PRO joins ideas from both PnP and RED, bridging between the two framework. Thus, RED-PRO provides flexibility in choosing the denoiser and is guaranteed to converge to stable solutions.
Furthermore, we present relaxed variants of RED-PRO related to the approximated fixed-point set, which is a much broader domain but not convex in general. Therefore, an open research direction is to provide the conditions under which the latter set is convex and derive different means to exploit it as a regularization. Finally, we perform experiments in image deblurring and super-resolution to demonstrate the validity of the regularization strategy, showing competitive and even improved performance in comparison to the original RED framework.

We conclude this section by briefly discussing key questions for future research. This study relies on the condition that the denoiser is demicontractive. Currently, determining whether an arbitrary denoiser satisfies this condition remains an open challenge. Certainly, any nonexpansive denoiser is also demicontractive,  yet, it is unclear how to verify the nonexpansiveness of a general function. Therefore, it is of utmost practical importance to derive an efficient technique for testing the demicontractivity of a given mapping. 
An alternative approach can be using learned mappings enforced to satisfy this condition. Similar concept is applied in the field of deep learning where neural networks are constrained by a given  Lipschitz constant \cite{gouk2018regularisation} to obtain regularization during training and robustness to adversarial attacks.

In addition, as in \cite{romano2017little}, the RED-PRO framework exploits denoisers that are designed for optimized distortion measure, such as PSNR. However, this might lead to a low perceptual quality of the resultant images \cite{blau2018perception}. Hence, an interesting research direction is the deployment of denoisers based on some perceptual loss. Employing such denoisers in the RED-PRO scheme might result in solutions that better address the perception-distortion tradeoff \cite{blau2018perception}. Finally, the family of demicontractive mappings covers a broad range of functions, hence, a promising study direction is incorporating within the RED-PRO framework powerful mechanisms, different from denoisers, which encapsulate prior knowledge on the unknown image.

\appendix

% \section{Proof of \cref{theo:averagedemi}}
% \label{app:averagedemi}
% First, it is clear that $\bigcap_i \fix(T_i)\subseteq \fix(T)$. To prove the reverse inclusion, let $\textbf{z}\in\bigcap_i \fix(T_i)$. Then, by \eqref{eq:demi2} we have 
% \begin{align*}
%      \langle T_i(\textbf{x})-\textbf{x},\,\textbf{x}-\textbf{z}\rangle&\leq-\frac{1-d}{2}\norm{T_i(\textbf{x})-\textbf{x}}^2\quad\forall \textbf{x}\in\bb{R}^n.
% \end{align*}
% Now, let $\textbf{x}\in\fix(T)$. Then, we get
% \begin{align*}
%     0&=\Langle T(\textbf{x})-\textbf{x},\,\textbf{x}-\textbf{z}\Rangle
%     \\ &=\sum_i w_i\Langle T_i(\textbf{x})-\textbf{x},\,\textbf{x}-\textbf{z}\Rangle
%     \\ &\leq -\sum_i w_i\frac{1-d}{2}\norm{T_i(\textbf{x})-\textbf{x}}^2.
% \end{align*}
% The above suggests that $\sum_i w_i\norm{T_i(\textbf{x})-\textbf{x}}^2=0$ which in turn implies that $T_i(\textbf{x})=\textbf{x}$ for all $i$, i.e., $\textbf{x}\in \bigcap_i \fix(T_i)$. Next, for any $\textbf{x}\in\bb{R}^n$ and $\textbf{z}\in\fix(T)$ it holds that
% \begin{align*}
%     \Langle \textbf{x}-T(\textbf{x}),\,\textbf{x}-\textbf{z}\Rangle&= \sum_i w_i\Langle \textbf{x}-T_i(\textbf{x}),\,\textbf{x}-\textbf{z}\Rangle \\
%     &\geq \sum_i w_i\frac{1-d_i}{2}\norm{\textbf{x}-T_i(\textbf{x})}^2 \\
%     &\geq \frac{1-d}{2} \sum_i w_i\norm{\textbf{x}-T_i(\textbf{x})}^2 \\
%     &\geq \frac{1-d}{2} \norm{\textbf{x}-T(\textbf{x})}^2,
% \end{align*}
% where the last inequality is due to \eqref{eq:inequality}. Hence, we conclude that $T$ is $d$-demicontractive.

\section{Proofs}
\label{app:proofs}

\subsection{Proof of \cref{prop:demi2quasi}}
\label{app:demi2quasi}
Given arbitrary $\textbf{x}\in\bb{R}^n$ and $\textbf{z}\in\fix(T)$, we have
\begin{align*}
    \norm{T_\alpha(\textbf{x})-\textbf{z}}^2&=\norm{(\textbf{x}-\textbf{z})+\alpha\Big(T(\textbf{x})-\textbf{x}\Big)}^2 \\
    &=\norm{\textbf{x}-\textbf{z}}^2-2\alpha\Langle \textbf{x}-\textbf{z},\textbf{x}-T(\textbf{x})\Rangle+\alpha^2\norm{\textbf{x}-T(\textbf{x})}^2.
\end{align*}
From \eqref{eq:demi2}, we get
\begin{equation*}
    \norm{T_\alpha(\textbf{x})-\textbf{z}}^2\leq \norm{(\textbf{x}-\textbf{z})}^2-\alpha(1-d-\alpha)\norm{T(\textbf{x})-\textbf{x}}^2.
\end{equation*}
Note that $\norm{T_\alpha(\textbf{x})-\textbf{x}}=\alpha\norm{T(\textbf{x})-\textbf{x}}$, implying that
\begin{equation*}
    \norm{T_\alpha(\textbf{x})-\textbf{z}}^2\leq \norm{(\textbf{x}-\textbf{z})}^2-\frac{(1-d-\alpha)}{\alpha}\norm{T_\alpha(\textbf{x})-\textbf{x}}^2,
\end{equation*}
which completes the proof.

\subsection{Proof of \cref{prop:demi2bound}}
\label{app:demi2bound}
For any bounded pair of points $\textbf{x},\textbf{z}\in[a,b]^n$, we have
\begin{equation*}
    \frac{1}{n}\norm{\textbf{x}-\textbf{z}}^2\leq (b-a)^2.
\end{equation*}
In addition, by \cref{prop:demi2quasi} we obtain
\begin{equation*}
    \norm{f_\alpha(\textbf{x})-\textbf{x}}^2\leq \frac{\alpha}{1-d-\alpha}\norm{\textbf{x}-\textbf{z}}^2,\;\forall \textbf{x}\in[a,b]^n,\,\textbf{z}\in\fix(f).
\end{equation*}
Hence, we get
\begin{equation*}
    \frac{1}{n}\norm{f_\alpha(\textbf{x})-\textbf{x}}^2\leq \frac{1}{n}\cdot\frac{\alpha}{1-d-\alpha}\norm{\textbf{x}-\textbf{z}}^2\leq \frac{\alpha}{1-d-\alpha}(b-a)^2,
\end{equation*}
which completes the proof.

\subsection{Proof of \cref{theo:diminishing}}
\label{app:diminishing}
Theorem\,5 of \cite{yamada2004hybrid} proves the convergence of \cref{algo:REDPROHSDM} to a point in \begin{equation*}
    \Gamma\triangleq\Big\{\textbf{x}\in \fix(f)\;|\;\ell(\textbf{x};y)=\underset{\textbf{x}\in \fix(f)}{\min}\;\ell(\textbf{x};y)\Big\},
\end{equation*}
under the assumptions that (A2) holds, $f_\alpha(\cdot)$ is quasi-nonexpansive quasi-shrinking mapping (\cite{censor1998interior}, Definition 4.3) and $\nabla\ell$ is paramonotone over $\fix(f)$, i.e.,
\begin{equation}
    \langle\nabla\ell(\textbf{u})-\nabla\ell(\textbf{v}), \textbf{u}-\textbf{v}\rangle = 0\;\Leftrightarrow\; \nabla\ell(\textbf{u})=\nabla\ell(\textbf{v}),\; \forall \textbf{u},\textbf{v}\in\fix(f).
    \tag{\cite{censor1998interior}, Definition 11}
\end{equation}
When (A1) is satisfied, then, $f_\alpha(\cdot)$ is continuous and strongly quasi-nonexpansive, hence, it is quasi-shrinking (See \cite{cegielski2014properties}, Corollary 4.2 and Proposition 4.4). Moreover, since $\ell(\cdot)$ is convex over the set $\fix(f)$, we have that $\nabla\ell$ paramonotone over $\fix(f)$ (See \cite{censor1998interior}, Lemma 12). Therefore, the conditions of Theorem\,5 in \cite{yamada2004hybrid} hold and the sequence $\{\textbf{x}_k\}_{k\in\bb{N}}$ is guaranteed to converge to an optimal solution of (\ref{eq:REDPRO}).

\subsection{Proof of \cref{theo:constant}}
\label{app:constant}
We start with the following lemma:
\begin{lemma}[Proposition 1 \cite{yamada2004hybrid}]
Let $T_1$ and $T_2$ be two strongly quasi-nonexpansive mappings with constants $\gamma_1$ and $\gamma_2$ respectively such that $\fix(T_1)\cap\fix(T_2)\neq \emptyset$. Then, the composition $T\triangleq T_1 \circ T_2$ is $\gamma$-strongly quasi-nonexpansive where $\gamma=\frac{\gamma_1\gamma_2}{\gamma_1+\gamma_2}$ and it satisfies $\fix(T)=\fix(T_1)\cap\fix(T_2)$. 
\label{lemma:composition}
\end{lemma}

\noindent By \cref{prop:demi2quasi}, the mapping $f_\alpha(\cdot)$ is strongly quasi-nonexpansive. In addition, for $
\mu\in (0, \frac{2}{L})$ we have that $G_\ell(\cdot)$ is averaged nonexpansive operator (See \cite{yamada2011minimizing}, Remark 17.16), in particular it is strongly quasi-nonexpansive. 
Define $T(\textbf{x})\triangleq f\Big(\textbf{x}-\mu\nabla\ell (\textbf{x})\Big)$, then by \cref{lemma:composition}, $T(\textbf{x})$ is continuous and $\gamma$-strongly quasi-nonexpansive for some $\gamma>0$ with $\fix(T)=\fix(f)\cap\fix(G_\ell)\neq\emptyset$. Furthermore, we can rewrite \cref{algo:REDPROHSDM} as $\textbf{x}_{k+1}=T(\textbf{x}_k)$. Thus, for any $\textbf{x}^*\in\fix(T)$ it holds that
\begin{equation*}
    \norm{\textbf{x}_{k+1}-\textbf{x}^*}^2=\norm{T(\textbf{x}_k)-\textbf{x}^*}^2\leq\norm{\textbf{x}_k-\textbf{x}^*}^2-\gamma\norm{\textbf{x}_{k+1}-\textbf{x}_k}^2\leq\norm{\textbf{x}_k-\textbf{x}^*}^2.
\end{equation*}
Therefore, the sequence $\{\textbf{x}_k\}_{k\in\bb{N}}$ is Fej\'er monotone w.r.t $\fix(T)$ and is bounded. Moreover, we can rewrite the above inequality as 
\begin{equation*}
 \norm{\textbf{x}_{k+1}-\textbf{x}_k}^2\leq\frac{1}{\gamma}\Big(\norm{\textbf{x}_k-\textbf{x}^*}^2-\norm{\textbf{x}_{k+1}-\textbf{x}^*}^2\Big)   
\end{equation*}.
Summing the above inequality for $k=0,...,t$ we get that
\begin{equation*}
    \sum_{k=0}^t \norm{\textbf{x}_{k+1}-\textbf{x}_k}^2 \leq \frac{1}{\gamma}\Big(\norm{\textbf{x}_0-\textbf{x}^*}^2-\norm{\textbf{x}_{t+1}-\textbf{x}^*}^2\Big) \leq   \frac{1}{\gamma}\norm{\textbf{x}_0-\textbf{x}^*}^2.
\end{equation*}
Thus, $\{\textbf{x}_k\}_{k\in\bb{N}}$ is a Cauchy sequence and $\textbf{x}_k\rightarrow \textbf{x}$ for some $\textbf{x}\in\bb{R}^n$. Since $\textbf{x}_{k+1}\rightarrow \textbf{x}$ and $\textbf{x}_{k+1}=T(\textbf{x}_k)\rightarrow T(\textbf{x})$, we have $\textbf{x}\in\fix(T)$ which in turn implies that $\textbf{x}\in\fix(f)$ and $\nabla\ell(\textbf{x})=0$, completing the proof.

\subsection{Proof of }
\label{app:Halpern}
The update rule for computing the projection onto the fixed-point set is
\begin{equation}
     \textbf{x}_{j+1}=f_\alpha\Big(t_j\textbf{x}_0+(1-t_j)\textbf{x}_j\Big).
\end{equation}
This can be rewritten in two stages as
\begin{align*}
    &\textbf{v}_{j+1}=t_j\textbf{x}_0+(1-t_j)\textbf{x}_j, \\
    &\textbf{x}_{j+1}=f_\alpha(\textbf{v}_{j+1}).
\end{align*}
Hence, we can rewrite the update rule with respect to $v_j$ as
\begin{equation*}
    \textbf{v}_{j+1}=t_j\textbf{x}_0+(1-t_j)f_\alpha(\textbf{v}_j).
\end{equation*}
Finally, by interchanging the roles of $x_j$ and $v_j$ we obtain \cref{eq:Halpern}.

\subsection{Proof of Theorem \ref{theo:dilation}}
\label{app:dilation}
Consider $\alpha\in(0,\frac{1-d}{2})$ and define the averaged operator $f_\alpha$. By \cref{prop:demi2quasi}, $f_\alpha$ is $\gamma$-strongly quasi-nonexpansive with $\gamma>1$. Hence, according to \cref{prop:displacemenbound} we have $\norm{\textbf{x}-f_\alpha(\textbf{x})}\leq\norm{\textbf{x}-P_{\fix(f)}(\textbf{x})}$. In addition, it holds that
\begin{equation*}
    \alpha\norm{\textbf{x}-f(\textbf{x})}=\norm{\textbf{x}-f_\alpha(\textbf{x})}.
\end{equation*}
Therefore, for any $\textbf{x}\in B_\delta(f)$ where $\delta\in[0,\alpha\epsilon]$, we have
\begin{align*}
    \norm{\textbf{x}-f(\textbf{x})}=\frac{1}{\alpha}\norm{\textbf{x}-f_\alpha(\textbf{x})} 
    \leq \frac{1}{\alpha}\norm{\textbf{x}-P_{\fix(f)}(\textbf{x})} 
    \leq \frac{\delta}{\alpha}\leq \epsilon,
\end{align*}
which implies that $\textbf{x}\in \fix_\epsilon(f)$.

%==================================================
%==================================================

\section{RED Revisited}
\label{subsec:REDrevisited}

Here, we complement the work in \cite{reehorst2018regularization} by considering the RED framework and studying the case in which the denoiser is non-differentiable. We do so by reformulating the RED optimization problem using the Rockafellar function \cite{rockafellar2009variational}. To make this part as self-contained as possible, we first recall essential results from variational analysis \cite{rockafellar2009variational}.

\begin{definition}
A continuous mapping $A$ is \textit{maximal cyclically monotone} if for every set of points $\{\textbf{x}_1,...,\textbf{x}_m\}$ in $\bb{R}^n$ and any $m\geq 2$ it holds that
\begin{equation*}
    \sum_{i=1}^m \Langle A(\textbf{x}_i),\textbf{x}_{i+1}-\textbf{x}_i)\Rangle \leq 0,\quad\text{where }\textbf{x}_{m+1}=\textbf{x}_1.    
\end{equation*}
\end{definition}

\begin{definition}
We say $f$ is \textit{cyclically firmly nonexpansive} if for every set of points $\{\textbf{x}_1,...,\textbf{x}_m\}$ in $\bb{R}^n$ and any $m\geq 2$ it holds that
\begin{equation*}
    \sum_{i=1}^m \Langle \textbf{x}_i-f(\textbf{x}_i),f(\textbf{x}_i)-f(\textbf{x}_{i+1})\Rangle \geq 0,\quad\text{where }\textbf{x}_{m+1}=\textbf{x}_1.      
\end{equation*}
\end{definition}

The following proposition connects cyclically firmly nonexpansive mappings with cyclically maximal monotone operators. This relationship shall allow us to establish the convergence of RED algorithms.
\begin{proposition}[\cite{bauschke2012firmly}]
Let $f:\bb{R}^n\rightarrow \bb{R}^n$ be cyclically firmly nonexpansive and consider the displacement mapping $A=Id-f$. Then, $A$ is maximal cyclically monotone.
\label{theo:cfne2mcm}
\end{proposition}

Now, we define the Rockafellar function which we shall use regularization. 
\begin{definition}[\textbf{Rockafellar function} \cite{rockafellar2009variational}] Consider a mapping $A:\bb{R}^n\rightarrow \bb{R}^n$ and let $\textbf{u}\in\bb{R}^n$ be an arbitrary point. For $m\geq2$, we define the following function with parameter $\textbf{u}$
\begin{equation*}
    C_A^m(\textbf{x};\textbf{u})\triangleq
    \begin{cases}
    \Langle \textbf{x},A(\textbf{u})\Rangle-\Langle \textbf{u},A(\textbf{u})\Rangle,\; &m=2. \\
    \underset{\{\textbf{x}_2,...,\textbf{x}_{m-1}\}}{\sup}\;\Langle \textbf{x}-\textbf{x}_{m-1},A(\textbf{x}_{m-1})\Rangle+\sum\limits_{i=1}^{m-2}\Langle \textbf{x}_{i+1}-\textbf{x}_i,A(\textbf{x}_i)\Rangle,\;&o.w.
    \end{cases}
\end{equation*}
where $\textbf{x}_1=\textbf{u}$.
Then, the Rockafellar function is defined as
\begin{equation*}
    R_A(\textbf{x};\textbf{u})\triangleq \underset{m\in\{2,3,...\}}{\sup}\;C_A^m(\textbf{x};\textbf{u}).
\end{equation*}
\end{definition}

\noindent The Rockafellar function in a well-known and important function in the area of variational analysis \cite{bartz2007fitzpatrick}, which exhibits favorable properties when $A$ is maximal cyclically monotone as stated below.

\begin{proposition}[Fact 3.2, \cite{bartz2007fitzpatrick}]
Let $A$ be maximal cyclically monotone and let $\textbf{u}\in\bb{R}^n$. Then, $R_A(\textbf{x};\textbf{u})$ is
convex, l.s.c, proper function with $R_A(\textbf{u};\textbf{u})=0$ and for any $\textbf{x}\in\bb{R}^n$
\begin{equation*}
    A(\textbf{x})=\nabla R_A(\textbf{x};\textbf{u}). 
\end{equation*}
Moreover, for any convex, l.s.c and proper function $h:\bb{R}^n\rightarrow\bb{R}$ for which $A=\nabla h$, it holds that $h(\textbf{x})=h(\textbf{u})+R_A(\textbf{x};\textbf{u}),\; \forall \textbf{x}\in\bb{R}^n$, i.e., $h(\cdot)$ is uniquely determined by $A$ up to an additive constant.
\label{prop:Rockafellarsubgradient}
\end{proposition}

Equipped with the above results, we formulate the next inverse problem. Let $f:\bb{R}^n\rightarrow\bb{R}^n$ be a continuous denoiser and define the denoiser residual $A\triangleq Id-f$. For arbitrary $\textbf{u}\in\bb{R}^n$ and $\lambda>0$, we consider the following minimization problem
\begin{equation}
    \underset{\textbf{x}\in\bb{R}^n}{\min}\; E_\text{Rock}(\textbf{x})=\ell(\textbf{x};\textbf{y})+\lambda R_A(\textbf{x};\textbf{u}).
    \label{eq:REDrock}
\end{equation}

When the denoiser $f$ is cyclically firmly nonexpansive, the above formulation admits convex minimization problem which can be solved by the RED algorithms presented in \cite{romano2017little}, as proven in the next theorem.   
\begin{theorem}
Assume that $f(\cdot)$ is cyclically firmly nonexpansive. Then, problem \eqref{eq:REDrock} is a convex optimization task where the gradient of the objective function is given by \eqref{eq:REDgradient}.  
\label{theo:rock}
\end{theorem}

\begin{proof}
By \cref{theo:cfne2mcm} in conjunction with \cref{prop:Rockafellarsubgradient}, the residual mapping $A=Id-f$ is maximal cyclically monotone, hence,  $R_A(\textbf{x};\textbf{u})$ is a convex differentiable function whose gradient is $\nabla R_A(\textbf{x};\textbf{u})=A(\textbf{x})=\textbf{x}-f(\textbf{x})$. The latter in turn implies that 
\begin{equation*}
    \nabla E_\text{Rock}(\textbf{x})=\nabla \ell(\textbf{x};\textbf{y})+\lambda\nabla R_A(\textbf{x};\textbf{u})=\nabla\ell(\textbf{x};\textbf{y})+\lambda\Big(\textbf{x}-f(\textbf{x})\Big)=\nabla E_\text{RED}(\textbf{x}).
\end{equation*}
\end{proof}
To conclude, \cref{theo:rock} provides the condition under which the RED approach minimizes a convex minimization problem given by \eqref{eq:REDrock}, and thus it proves the convergence of RED algorithms. This holds even for the case where the denoiser is not differentiable. We note that the requirement of $A=Id-f$ to be maximal cyclically monotone is a necessary condition, whereas the demand of $f(\cdot)$ to be cyclically firmly nonexpansive is a sufficient condition that might be relaxed and is fulfilled by e.g. proximal operators.

% \section{Proof of \cref{theo:composition}}
% \label{app:composition}
% It is easy to see that $\fix(T_1)\cap\fix(T_2)\subseteq \fix(T)$, hence, we prove the reverse inclusion. Let $x\in \fix(T)$ and $y\in \fix(T_1)\cap\fix(T_2)$. First, let us assume by contradiction that $x\notin\fix(T_2)$ and $T_2(x)\notin\fix(T_1)$. Then,
%     \begin{equation*}
%         \norm{x-y}=\norm{T_1\Big(T_2(x)\Big)-y}<\norm{T_2(x)-y}<\norm{x-y},
%     \end{equation*}
%     which is a contradiction. Therefore, we consider the following two cases:
% \begin{enumerate}
%     \item $x\in\fix(T_2)$. Then, $T_1(x)=T_1\Big(T_2(x)\Big)=T(x)=x$, i.e., $x\in \fix(T_1)\cap\fix(T_2)$.
%     \item $T_2(x)\in\fix(T_1)$. Then, $T_2(x)=T_1\Big(T_2(x)\Big)=T(x)=x$ which indicates that $x\in \fix(T_1)\cap\fix(T_2)$.
% \end{enumerate}
% Now, notice that for all $x\in\bb{R}^n$ and $z\in\fix(T)$, we have
% \begin{align*}
%     \norm{T(x)-z}^2&=\norm{T_1\Big(T_2(x)\Big)-z}^2 \\
%     \leq \
% \end{align*}

\newpage
\bibliographystyle{siamplain}
\bibliography{refs}

\end{document}